\documentclass{IEEEtran}
\usepackage{cite}
\usepackage{amsmath,amssymb,amsfonts}
\usepackage{algorithmic}
\usepackage{graphicx}
\usepackage{textcomp}
\def\BibTeX{{\rm B\kern-.05em{\sc i\kern-.025em b}\kern-.08em
    T\kern-.1667em\lower.7ex\hbox{E}\kern-.125emX}}

\usepackage{xcolor}
\usepackage{mathtools}
\usepackage{amsfonts}
\usepackage{tensor}
\usepackage{hyperref}
\usepackage{amsthm}

\usepackage{tikz}
\usepackage{wrapfig}

\DeclareMathOperator{\vect}{vec}

\DeclareMathOperator{\Diag}{Diag}

\DeclareMathOperator{\median}{median}

\newtheorem{theorem}{Theorem}[section]
\newtheorem{remark}[theorem]{Remark}
\newtheorem{definition}[theorem]{Definition}
\newtheorem{lemma}[theorem]{Lemma}

\DeclareMathOperator{\Exs}{\mathbb{E}}

\usetikzlibrary{shapes,arrows,fit,calc,positioning,automata}
\tikzstyle{int}     = [draw, minimum size=4em]
\tikzstyle{init}    = [pin edge={to-,thin,black}]
\tikzstyle{ADD} = [draw,circle]
\tikzstyle{line}    = [draw, -latex']

\tikzstyle{box}     = [draw, minimum size=4em]
\tikzstyle{POINT} = [draw, circle,fill,scale=0.3]

\usepackage[ruled, noend, noline]{algorithm2e}
\SetKwComment{Comment}{$\triangleright$\ }{}

\makeatletter
\newcommand{\removelatexerror}{\let\@latex@error\@gobble}
\makeatother

\begin{document}
\title{Randomized Polar Codes for Anytime Distributed Machine Learning}
\author{Burak Bartan and Mert Pilanci, \IEEEmembership{Member, IEEE}
\thanks{This work was supported by the National Science Foundation under grants NSF CAREER CCF-2236829, ECCS-2037304, DMS-2134248, Army Research Office Early Career Award W911NF-21-1-0242, ACCESS – AI Chip Center for Emerging Smart Systems sponsored by InnoHK funding, Hong Kong SAR, and a Precourt Pioneering Projects seed grant. %We would like to thank the anonymous reviewers for their suggestions that improved our paper. 
}
\thanks{Burak Bartan and Mert Pilanci are with the Department of Electrical Engineering, Stanford University, Stanford, CA 94305 (e-mail: bbartan@stanford.edu; pilanci@stanford.edu).}
% \thanks{This paper has supplementary downloadable material available at http://ieeexplore.ieee.org, provided by the author. The material includes the appendix.}
}

\maketitle
\begin{abstract}
We present a novel distributed computing framework that is robust to slow compute nodes, and is capable of both approximate and exact computation of linear operations. The proposed mechanism integrates the concepts of randomized sketching and polar codes in the context of coded computation. We propose a sequential decoding algorithm designed to handle real valued data while maintaining low computational complexity for recovery. Additionally, we provide an anytime estimator that can generate provably accurate estimates even when the set of available node outputs is not decodable. We demonstrate the potential applications of this framework in various contexts, such as large-scale matrix multiplication and black-box optimization. We present the implementation of these methods on a serverless cloud computing system and provide numerical results to demonstrate their scalability in practice, including ImageNet scale computations.
\end{abstract}

\begin{IEEEkeywords}
Polar codes, distributed algorithms, randomized sketching, machine learning, large-scale computing.
\end{IEEEkeywords}

\section{Introduction} \label{sec:introduction}

The utilization of distributed computing has become a crucial aspect in various scientific and engineering applications that involve the manipulation of large-scale data and models. Despite its advantages, distributed computing poses several challenges in algorithm design, including inter-node communication, handling of malfunctioning or slow nodes, and maintaining data privacy.

In distributed computing, the presence of straggling nodes can significantly impact the overall computation time of an algorithm. To overcome the problem of stragglers, the idea of adding redundancy to computations using error-correcting codes has been explored in the literature by many recent works \cite{Lee2018}, \cite{baharav2018prodcodes}, \cite{yu2017polycode}. Error-correcting codes not only help speed up computations by making it possible to compute the desired output without waiting for the outputs of the straggling workers, but also provide resilience against crashes and timeouts, leading to a more robust distributed computing framework.

In this study, we propose a novel method that incorporates the principles of anytime computing \cite{dean88anytime} with coded computation and randomized sketching. The anytime computing approach allows for the generation of approximate solutions that improve in accuracy over time. Our method utilizes this concept by enabling exact recovery when the available outputs are decodable. In situations where the available outputs are not decodable, our method can provide accurate unbiased estimates for the desired computation at any point in time. These anytime estimates can be applied in machine learning and optimization as they can provide approximate gradients for these applications.

Our approach incorporates the use of polar codes for computation \cite{bartan2019straggler} and randomized sketching to achieve provable improvements in performance. Polar coding, a method for error-correcting code construction, has been shown to achieve the capacity of symmetric binary-input discrete memoryless channels \cite{polar2009arikan}. As a result, we propose a flexible distributed computing method that is highly robust against stragglers and able to provide accurate approximations when exact recovery fails due to the high number of stragglers. The formal statement of our main result could be found in Theorem \ref{thm:main_result}.

Serverless computing is a novel cloud-based computational model that enables users to execute computations within the cloud environment without the need for provisioning or managing servers. In order to ensure straggler resilience within this model, it is essential for the scheme to be scalable, as the number of nodes may vary greatly. As the number of worker nodes increases to the scale of hundreds or thousands, two considerations become particularly salient. The first one is that encoding and decoding of the code must be low complexity. The second one is that one must be careful with the numerical round-off errors if the inputs are not from a finite field, but instead are full-precision real numbers. To clarify this point, when the inputs are real-valued, encoding and decoding operations introduce round-off errors. Polar codes show superiority over many codes in terms of both of these aspects as we will show in the sequel. They have low encoding and decoding complexity and both encoding and decoding require only a small number of subtraction and addition operations without any multiplication operations. In addition, they are known to achieve channel capacity in communication (\cite{polar2009arikan}). The importance of this fact for coded computation is that we assume the outputs of worker nodes are analogous to binary erasure channels and thus the number of worker outputs needed for decoding is asymptotically optimal.

The distinction between serverless and server-based computing is crucial in understanding the utility of polar coding based approaches. Serverless computing requires a larger number of compute nodes to perform equivalent computation compared to server-based computing. This design consideration highlights the need for efficient encoding and decoding algorithms to address the increased computational requirements. For instance, for a distributed server-based system with $N=8$ nodes, using maximum distance separable (MDS) codes with decoding complexity as high as $O(N^3)$ can be computationally feasible. However, this becomes a limitation when transitioning to a serverless system, as the number of functions $N$ required to achieve the same level of computation may increase to several thousand. The reason for this is the limited resources such as low RAM and short lifetime that each function is restricted to have in serverless computing. Hence, it is imperative to employ codes with efficient decoding algorithms, such as polar codes. Despite lacking the MDS properties, polar codes exhibit a diminishing performance gap in terms of recovery threshold for large code block-lengths.
\subsection{Related Work} \label{subsec:related_work}

Coded matrix multiplication has been introduced in \cite{Lee2018} for speeding up distributed matrix-vector multiplication in server-based computing platforms. In \cite{Lee2018}, it was shown that it is possible to speed up distributed matrix multiplication by using MDS codes. MDS codes however have the disadvantage of having high encoding and decoding complexity, which could be restricting in setups with large number of workers. The work in \cite{baharav2018prodcodes} attacks this problem by introducing a coded computation scheme based on $d$-dimensional product codes. \cite{yu2017polycode} presents a scheme referred to as polynomial codes for coded matrix multiplication with input matrices from a large finite field. This approach might require quantization for real-valued inputs which could introduce additional numerical issues. \cite{dutta2018optimalrec} and \cite{yu2018straggler} are other works investigating coded matrix multiplication and provide analysis on the optimal number of worker outputs required. The polynomial code approach has been extended in \cite{yu2020entangled}, where a secure, private, batch distributed matrix multiplication scheme has been proposed. 

The matrix multiplication computation is partitioned into sequentially computed layers of varying precision in \cite{esfahanizadeh2022layered}. Coded computing is then applied to each precision layer to make the computation of each layer straggler-resilient. \cite{yang2019timely} considers the variability in the computation speed across the worker nodes. The state of each worker node is modeled using Markov chains and a dynamic computation strategy is developed. This work also has a coded computing aspect which is based on an MDS code. Our polar code approach could be applied to \cite{yang2019timely} to lower the complexity required for decoding in order to enable it for large-scale distributed computing. A common ingredient in both \cite{esfahanizadeh2022layered} and \cite{yang2019timely} is the use of finite field data and high complexity decoding (typically cubic complexity). Fast encoding and decoding procedures along with the approximation property of our method enable massive scale computing. \cite{jahani2023berrut} proposes an approximate coded computing algorithm with low complexity that act on real-valued data. Authors of \cite{jahani2023berrut} provide theoretical guarantees for the approximation quality. Our method not only exhibits similar desirable properties but also allows for fast exact decoding.

There are works in the literature on coded computation for gradient coding and different types of large-scale linear algebra operations such as \cite{wang19gradient_coding}, \cite{wang19coded_linear}, \cite{shashanka17lowrank}. Straggler mitigation in distributed computing with heterogeneous compute nodes is studied in \cite{reisizadeh17coded}. In addition to the coding theoretic approaches, \cite{gupta2018oversketch} presents an approximate straggler-resilient matrix multiplication scheme where sketching and straggler-resilient distributed matrix multiplication are combined.

Polar codes for coded computation is further studied in the work \cite{pilanci2021comppolarization}. This work analyzes the convergence properties of polarization of computation times.

Using Luby Transform (LT) codes, a type of rateless fountain codes, in coded computation has been proposed in \cite{severinson2018lt} and \cite{mallick2018rateless}. The proposed scheme in \cite{mallick2018rateless} divides the overall task into smaller tasks of multiplication of rows of $A$ with $x$ for better load-balancing. The work \cite{severinson2018lt} proposes the use of inactivation decoding and the work \cite{mallick2018rateless} uses peeling decoder. Peeling decoder has a computational complexity of $O(N\log N)$, however its performance is not satisfactory if the number of input symbols is not very large. Inactivation decoder performs better than the peeling decoder in terms of error correction, however, it is not as fast as the peeling decoder.

Several alternative methods were proposed for designing anytime algorithms in distributed computing, capable of generating estimates and exact solutions with sufficient time (\cite{miguel2016anytime, ferdinand2016anytime, ferdinand2017anytime}). However, it is not guaranteed that an algorithm that provides exact recovery is computationally efficient. Our use of a randomized version of polar codes, which have efficient decoding algorithms as well as strong concentration properties, makes this possible. Our approach is novel in that it enables efficient exact recovery while also offering anytime inexact solutions with theoretical guarantees as outlined in Theorem \ref{thm:main_result}.

Among the recent work on serverless computing for machine learning training is \cite{carreira2018caseserverless}, where serverless machine learning is discussed in detail, and challenges and possible solutions on serverless machine training are provided. Resource allocation and pricing aspects of serverless computing are investigated in \cite{gupta20utility}. Authors in \cite{feng2018serverlesstraining} consider an architecture with multiple master nodes and state that for small neural network models, serverless computing helps speed up hyperparameter tuning. Similarly, the work in \cite{wang2019serverlesslearning} shows via experiments that their prototype on AWS Lambda can reduce model training time greatly. A distributed convex optimization mechanism based on randomized second order optimization is proposed and studied in \cite{gupta20oversketched} for serverless computing.

We also investigate the application of black-box optimization methods for reinforcement learning. The work of \cite{salimans2017es} considers the evolution strategies method in reinforcement learning and show that distributed training with evolution strategies can be very fast because of its scalability. The work of \cite{choromanski2018structured} shows that using orthogonal exploration directions leads to lower errors and present the structured evolution strategies method which is based on a special way of generating random orthogonal exploration directions, as we discuss later in detail.
\subsection{Overview of Our Contributions} \label{subsec:contributions}

\begin{itemize}
    \item We introduce a novel approach for distributed computation of linear operations that is resistant to slow or ``straggler" nodes while also providing approximate solutions when exact decoding is not possible. The method unifies polar codes and randomized Hadamard sketches to achieve this goal. This allows for robustness and flexibility in computation and highly efficient fast decoding, making it a versatile solution for large scale linear operation computations.
    
    \item We present methods for coded computation and black-box optimization using polar codes. We develop efficient algorithms for encoding and decoding over real numbers.
    
    \item We have extended the previous results on polarization of computation times into kernels of arbitrary size. Arbitrary size kernels could be useful when the number of nodes is not a power of $2$.
    
    \item We present numerical results on large-scale data including ImageNet (\cite{imagenet}) that show the scalability of the proposed methods. We have implemented and tested the methods for the serverless computing service AWS Lambda.
\end{itemize}

\section{Coded Computation using Polar Codes} \label{sec:coded_comp}

\subsection{Problem Setup}
In this section, we will use the distributed computation of the matrix-vector multiplication operation $A x$ as a motivating example for the discussion. Suppose that $A \in \mathbb{R}^{n \times d}$ is a large data matrix partitioned to $s$ sub-matrices $A_i$ of size $\frac{n}{s}$-by-$d$ over its rows:
\begin{align}
    A = \begin{bmatrix} A_1 \\ \vdots \\ A_s \end{bmatrix} \mbox{, where } A_i \in \mathbb{R}^{
    \frac{n}{s} \times d}, \quad i = 1,\dots,s.
\end{align}
To keep the presentation simple, we will assume that $x\in \mathbb{R}^{d}$ is a vector of manageable size and is not partitioned or encoded. For a detailed discussion of the general setting where $x$ is a matrix and also encoded, the reader is referred to the appendix. Our goal is to compute exactly or approximately the product $A x$ using $N$ worker nodes that run in parallel. In our model, worker nodes are allowed to communicate only with the central node. We will assume the encoded data blocks are denoted as $\tilde{A}_i$:
\begin{align}
    \tilde{A} = \begin{bmatrix} \tilde{A}_1 \\ \vdots \\ \tilde{A}_N \end{bmatrix} , \mbox{ where } \tilde{A}_i \in \mathbb{R}^{\frac{n}{s}\times d}, \quad i=1,\dots,N.
\end{align}
The output of worker node $i$ is then $\tilde{A}_i x$. Figure \ref{fig:dist_computing} presents a visual representation of the computing model. When worker node $i$ finishes its assigned computation, we say that its output is available. The set of nodes whose outputs are available will be denoted as $\mathcal{S} \subseteq \{1,\dots,N\}$.  We will assume that the worker node outputs can be either ``unavailable" or ``available and correct". In other words, if an output is available, we will assume it is error-free. Consequently, the worker nodes can be modeled as real-valued erasure channels.

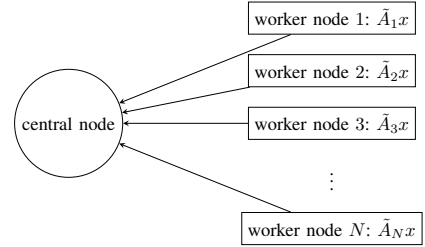
\begin{figure}
\centering
\scalebox{0.70}{
\begin{tikzpicture}[node distance={10mm}, main/.style = {draw, circle}, squarednode/.style={rectangle, draw=black, minimum size=4mm}, dotnode/.style={draw,shape=circle,fill=black,inner sep=0pt,minimum size=1.5mm}]

\node[squarednode] (1) {worker node $1$: $\tilde{A}_1x$};
\node[squarednode] (2) [below of=1] {worker node $2$: $\tilde{A}_2x$};
\node[squarednode] (3) [below of=2] {worker node $3$: $\tilde{A}_3x$};
\node[] (4) [below of=3] {$\vdots$};
\node[squarednode] (5) [below of=4] {worker node $N$: $\tilde{A}_Nx$};

\node[] (6) [left of=3] {};
\node[] (7) [left of=6] {};
% \node[] (8) [left of=7] {};
\node[] (10) [left of=7] {};
\node[] (11) [left of=10] {};
\node[main] (9) [left of=11] {central node};

\draw[stealth-] (9) -- (1);
\draw[stealth-] (9) -- (2);
\draw[stealth-] (9) -- (3);
\draw[stealth-] (9) -- (5);

\end{tikzpicture}
}
\vspace{-2mm}
\caption{Distributed computing model.}
\label{fig:dist_computing}
\end{figure}

\subsection{Main Result}

Our primary contribution is a novel technique for introducing redundancy in computation that effectively eliminates the straggler effect in exact recovery and at the same time, provides approximation guarantees for the anytime estimates. Our scheme enables finding an unbiased estimator for the matrix-vector product $Ax$ that provides guaranteed approximation results when the set of available outputs is not sufficient for exact decoding of the result. An essential aspect of our analysis is the synthesis of two distinct areas of research: Polar codes and the concentration of measure for randomized Hadamard sketches.

The proposed method is summarized in Algorithm \ref{alg:main_alg}. Our encoder acts on the blocks $A_i$ of the data matrix $A$ and returns the encoded data blocks $\tilde{A}_i$. Details of the encoding procedure are given in Section \ref{sec:encoding}. The worker nodes are assigned tasks such that worker node $i$ computes the product $\tilde{A}_ix$. Then, the central node monitors the available node outputs and updates the set $\mathcal{S}$ and approximation of the computation accordingly. Once $\mathcal{S}$ becomes decodable (see Definition \ref{def:decodable}), the decoder recovers the desired result $Ax$ exactly. We will provide the details of the encoding and decoding algorithms in the sequel.

\begin{definition}[Decodable set] \label{def:decodable}
Let $\mathcal{S}$ denote the set of nodes that complete their assigned computation. We say $\mathcal{S}$ is decodable whenever it is possible to exactly recover the desired result from the outputs of nodes in $\mathcal{S}$ using our sequential decoding procedure (Algorithm \ref{decoding_alg}).
\end{definition}

\DontPrintSemicolon
\begin{algorithm}
 \KwIn{Data $A, x$}
 
 $\tilde{A} = $ encoder($A$) (see Section \ref{subsec:randpolarcode})\;
 Assign $\tilde{A}_ix$ to worker node $i$, $i=1,\dots,N$ \;
 
 Initialize $\mathcal{S} = \{\}$
 
 \While{$\mathcal{S}$ not decodable}{
 update $\mathcal{S}$ \;
 form the approximate anytime estimator $\mathcal{T}_x(\mathcal{S}) := \vect\left(\frac{1}{|\mathcal{S}|} \sum_{i \in \mathcal{S}} z_i(\tilde{A}_ix)^T \right)$\;
 }
 
 $Ax = $ decoder($\mathcal{S}$) (see Algorithm \ref{decoding_alg}) \;
 return $Ax$ (exact computation) \;
 \caption{Randomized polar coding based anytime distributed computing}
 \label{alg:main_alg}
\end{algorithm}

We defer the analysis of exact decodability via sequential decoding to Section \ref{sec:comp_times}. When $\mathcal{S}$ is not decodable, we will show that it is possible to construct approximate solutions using the estimator $\mathcal{T}_x(\mathcal{S})$ defined as
\begin{align} \label{eq:estimator}
    &\mathcal{T}_x(\mathcal{S}) := \vect\left(\frac{1}{|\mathcal{S}|} \sum_{i \in \mathcal{S}} z_i(\tilde{A}_ix)^T \right) \in \mathbb{R}^n
\end{align}
where $z_1,...,z_{|\mathcal{S}|} \in \mathbb{R}^s$ are column vectors corresponding to the $|\mathcal{S}|$ rows of the randomized polar code matrix $Z=[z_1,...,z_N]^T$ as defined in Section \ref{subsec:randpolarcode} that correspond to the nodes that complete their assigned computation. In the above formula, the $\vect(\cdot)$ operation vectorizes the $s\times\frac{n}{s}$ matrix $\frac{1}{|\mathcal{S}|} \sum_{i \in \mathcal{S}} z_i(\tilde{A}_ix)^T$ columnwise into a size-$n$ vector to provide an anytime estimator for $Ax \in \mathbb{R}^n$.

We now present our main result on the quality of the anytime estimator.

\begin{theorem}[Main result] \label{thm:main_result}
Suppose that the worker run times are independently distributed. Algorithm \ref{alg:main_alg} returns the matrix-vector product $Ax$ exactly for decodable\footnote{The distribution of the time at which exact computation is possible can be found in closed form as we show in Section \ref{sec:comp_times}} $\mathcal{S}$ in $O(N\log N)$ decoding time. In the case when $\mathcal{S}$ is not decodable, the approximation $\eqref{eq:estimator}$ is unbiased and for any fixed collection of vectors $x,x^\prime,c$ we have
\begin{align}
(1-\epsilon) \|A(x-x^\prime)\|_2^2 &\le \langle \mathcal{T}_x(\mathcal{S})-\mathcal{T}_{x^\prime}(\mathcal{S}), A(x-x^\prime) \rangle \nonumber \\
&\qquad \qquad \le (1+\epsilon)\|A(x-x^\prime)\|_2^2, \label{eq:xxprimeapproximation}
\end{align}
and
\begin{align}
-\epsilon (\|c\|_2^2+\|Ax\|_2^2) &\le 
 \langle c, \mathcal{T}_x(\mathcal{S})-Ax\rangle \le \epsilon(\|c\|_2^2+\|Ax\|_2^2), \label{eq:xxprimelinearapproximation}
\end{align}
with probability at least $1-s\exp{(-C_2|\mathcal{S}|)}$ when at least $C_1\log(N)^4/\epsilon^2$ workers finish their computation\footnote{$C_1,C_2$ are constants independent of dimensions.}. Here, $\mathcal{T}_x(\mathcal{S})$ and $\mathcal{T}_{x^\prime}(\mathcal{S})$ are approximations of $Ax$ and $Ax^\prime$ respectively. In addition, the estimator $\mathcal{T}_x(\mathcal{S})$ can be computed in $O(N\log N)$ time.
\end{theorem}

The proof is presented in Sections \ref{subsec:proof} and \ref{subsec:exactrec}. The bound \eqref{eq:xxprimelinearapproximation} shows that the inner-product $c^TAx$ can be replaced with $c^T \mathcal{T}_x(S)$ with small error for any fixed vector $c$. A corollary of the approximation in \eqref{eq:xxprimeapproximation} is that $1-\epsilon\le \langle \mathcal{T}_x(\mathcal{S}), Ax \rangle \le 1+\epsilon$ for any $x$ satisfying the normalization $\|Ax\|_2=1$. The main significance of Theorem \ref{thm:main_result} is that it shows there exists a strategy of encoding data that allows for unbiased estimates at anytime and eventual exact recovery. Moreover, the complexity of encoding and decoding with exact 
recovery, and also approximate recovery is only $O(N\log N)$, which is significantly faster than standard computational codes \cite{Lee2018}.

%%%%%%%%%%%%%
\subsection{Encoding}
\label{sec:encoding}
Here we describe the randomized polar coding construction underlying our computational scheme. We will use $U_1,\dots,U_N$ to denote the input data blocks to the encoder. Some of the inputs will be set to zero and others will be set to the data blocks $A_i$. This is to control the redundancy of the code, analogous to the frozen and data bits in traditional polar codes. 

Consider the construction given in Figure \ref{fig:circuit_fig_N_2} for $N=2$ data blocks. The rectangles $M_1$ and $M_2$ correspond to worker nodes $1$ and $2$. The input blocks $U_1$ and $U_2$ are first multiplied by the diagonal matrix $D=\Diag(D_1, D_2)$ where the diagonal entries $D_i$ are sampled from a scalar Rademacher distribution with $\pm 1$ i.i.d. entries. This is then input to the Hadamard kernel. The outputs are the encoded data blocks and equal to $D_1U_1+D_2U_2$ and $D_1U_1-D_2U_2$. The worker nodes are tasked with multiplying these encoded blocks by $x$. The variables $D_i$ are generated only once during the construction of the code and kept fixed during the decoding.

\begin{remark}
Note that the encoding shown in Figure \ref{fig:circuit_fig_N_2} is for the Hadamard kernel $\bigl[\begin{smallmatrix} 1 & ~\,1 \\ 1 & -1 \end{smallmatrix}\bigr]$, whereas the polar code kernel is given by $\bigl[\begin{smallmatrix} 1 & ~1 \\ 0 & ~1 \end{smallmatrix}\bigr]$. Encoding and exact decoding procedures work  similarly for both kernels. However, approximate recovery is only possible via the Hadamard kernel due to the connection to randomized Hadamard sketches as shown in the sequel.
\end{remark}

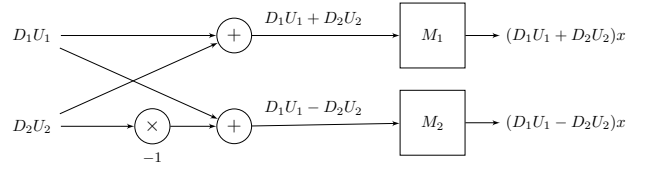
\begin{figure}
\centering
\scalebox{0.74}{
\begin{tikzpicture}[node distance=\nodedistance,auto,>=latex', scale = 0.83, transform shape]
    \tikzstyle{line}=[draw, -latex']
    \node [int] (M1) { $M_1$ };
    \node [int, below=0.47 of M1] (M2) { $M_2$ };
    \node [ADD,left=3.2 of M1] (ADD1) {\large $+$};
    \node [ADD,below=1.2 of ADD1] (ADD2) {\large $+$};
    \node [ADD,left=1 of ADD2] (NEG) {\large $\times$};
    \node [below=0.05 of NEG] (MINUS) {\small $-1$};
    \node [above=0.1 of $(ADD1)!0.4!(M1)$] (MINUS) {$D_1U_1+D_2U_2$};    
    \node [above=0.1 of $(ADD2)!0.4!(M2)$] (MINUS) {$D_1U_1-D_2U_2$};
    \node [right=0.75 of M1] (f1) {$(D_1U_1 + D_2U_2)x$};      
    \node [right=0.75 of M2] (f2) {$(D_1U_1 - D_2U_2)x$};      
    \node [left of=ADD1,node distance=5.0cm,text width=1cm,anchor=west,align=center] (A1) {$D_1 U_1$};
    \node [left of=ADD2,node distance=5.0cm,text width=1cm,anchor=west,align=center] (A2) {$D_2U_2$};
    
    \path[line] (ADD1) edge (M1)
                (ADD2) edge (M2)
                (A1) edge (ADD1)
                (A2) edge (NEG)
                (NEG) edge (ADD2)
                (A2) edge (ADD1)
                (A1) edge (ADD2)
                (M1) edge (f1)
                (M2) edge (f2);
\end{tikzpicture}
}
\vspace{-3mm}
\caption{Randomized polar code construction for $N=2$.}
\label{fig:circuit_fig_N_2}
\end{figure}

Larger size constructions can be obtained by recursively applying the size $2$ constructions. For instance, the construction for $N=4$ is shown in Figure \ref{fig:circuit_fig_N_4}. The recursive construction enables fast encoding and decoding operations, making it suitable for large-scale computing.

We now discuss the procedure for determining which inputs $U_i$ to freeze, i.e. set to zero, and which ones to send in data blocks. Let us denote the erasure probability of each worker node by $\epsilon$, and assume that the erasures are independent. The calculation of the erasure probabilities for the transformed nodes is similar to traditional polar codes, which we show with the example of $N=2$: Suppose that a sequential decoder in the first stage recovers $D_1U_1$ from the outputs of $M_1$ and $M_2$ in Figure \ref{fig:circuit_fig_N_2}, and then recovers $D_2U_2$ given $D_1U_1$ in the second stage. The probability that the first stage fails is given by
\begin{align*}
&\mathbb{P}[\textrm{output $M_1$ is erased} \textbf{ or } \textrm{output $M_2$ is erased} ]  = \\
&\qquad = F_+(\epsilon):= 1-(1-\epsilon)^2,
\end{align*}
since any erasure makes the recovery of $D_1U_1$ impossible.
On the other hand, the probability that the second stage fails is given by
\begin{align*}
&\mathbb{P}[\textrm{output $M_1$ is erased} \textbf{ and } \textrm{output $M_2$ is erased} ] = \\
&\qquad = F_-(\epsilon):= \epsilon^2,
\end{align*}
since $D_2U_2$ can be recovered either from the output of $M_1$ or $M_2$ given the knowledge of $D_1U_1$.
For larger construction sizes, the above calculation can be extended recursively. For $N=4$, the erasure probabilities of the transformed nodes are given by $$\Big\{F_+(F_+(\epsilon)), F_+(F_-(\epsilon)), F_-(F_+(\epsilon)), F_-(F_-(\epsilon))\Big\},$$ (see e.g., \cite{bartan2019straggler}).
Based on the erasure probabilities of the transformed nodes, we select the best ones for data, and freeze the rest (i.e., set to zero matrices). These transformed nodes are analogous to virtual channels in polar coding for communication \cite{polar2009arikan}.

After computing the erasure probabilities for the transformed nodes, we choose the $N (1-\epsilon)$ nodes with the lowest erasure probabilities as data nodes. The remaining $N\epsilon$ nodes are frozen. For example, for $N=4$ and $\epsilon = 0.5$, the erasure probabilities of the transformed nodes are calculated to be $\{0.938, 0.563, 0.438, 0.063\}$. It follows that we freeze the first two inputs, and the last two inputs are set to data blocks. This means that in Figure \ref{fig:circuit_fig_N_4}, we set $U_1=U_2=0^{n/2\times d}$ and $U_3=A_1$, $U_4=A_2$.

Note that unlike the XOR operation of polar codes in binary communication channels, in this work we consider real numbers and linear polarizing transformations over real numbers. As it will be shown in the sequel, the channel polarization phenomenon in finite fields carries over to the reals in an analogous manner. We note that the encoding procedure has computational complexity $O(N\log N)$. This can be seen by observing that there are $\log_2 N + 1$ vertical levels in the code construction and $N$ nodes in every level.

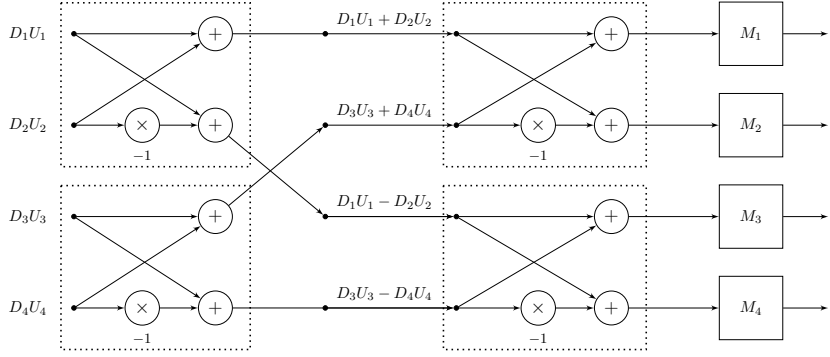
\begin{figure*}
\centering
\scalebox{0.75}{
\begin{tikzpicture}[node distance=\nodedistance,auto,>=latex', scale = 0.8, transform shape]
    \tikzstyle{line}=[draw, -latex']

%% operations    
    \node [ADD] (ADD1) {\large $+$};
    \node [ADD,below=1.26 of ADD1] (ADD2) {\large $+$};
    \node [ADD,left=0.85 of ADD2] (NEG) {\large $\times$};
    \node [ADD,below=1.26 of ADD2] (ADD3) {\large $+$};
    \node [ADD,below=1.26 of ADD3] (ADD4) {\large $+$};
    \node [ADD,left=0.85 of ADD4] (NEG2) {\large $\times$};    
    \node [below=0.05 of NEG] (MINUS) {\small $-1$};
    \node [below=0.05 of NEG2] (MINUS) {\small $-1$};
%% points
    \node [POINT, right=2 of ADD1] (S12) {};
    \node [POINT, right=2 of ADD2] (S22) {};
    \node [POINT, right=2 of ADD3] (S32) {};
    \node [POINT, right=2 of ADD4] (S42) {};
%% points
    \node [POINT, right=2.8 of S12] (S1) {};
    \node [POINT, right=2.8 of S22] (S2) {};
    \node [POINT, right=2.8 of S32] (S3) {};
    \node [POINT, right=2.8 of S42] (S4) {};
%% second layer operations
    \node [ADD, right=3 of S1] (ADD12) {\large $+$};
    \node [ADD, right=3 of S2] (ADD22) {\large $+$};
    \node [ADD, right=3 of S3] (ADD32) {\large $+$};
    \node [ADD, right=3 of S4] (ADD42) {\large $+$};
    \node [ADD,left=0.85 of ADD22] (NEG21) {\large $\times$};
    \node [ADD,left=0.85 of ADD42] (NEG22) {\large $\times$};
    \node [below=0.05 of NEG21] (MINUS) {\small $-1$};
    \node [below=0.05 of NEG22] (MINUS) {\small $-1$};
%% input      
    \node [POINT,left=2.7 of ADD1] (A1p) {};
    \node [POINT,left=2.7 of ADD2] (A2p) {};
    \node [POINT,left=2.7 of ADD3] (A3p) {};
    \node [POINT,left=2.7 of ADD4] (A4p) {};
    \node [left of=ADD1,node distance=4.7cm,text width=0.5cm,anchor=west,align=center] (A1) {$D_1U_1$};
    \node [left of=ADD2,node distance=4.7cm,text width=0.5cm,anchor=west,align=center] (A2) {$D_2U_2$};
    \node [left of=ADD3,node distance=4.7cm,text width=0.5cm,anchor=west,align=center] (A3) {$D_3U_3$};
    \node [left of=ADD4,node distance=4.7cm,text width=0.5cm,anchor=west,align=center] (A4) {$D_4U_4$};
%% machines
    \node [box, right=2 of ADD12] (M1) {$M_1$};    
    \node [box, right=2 of ADD22] (M2) {$M_2$};    
    \node [box, right=2 of ADD32] (M3) {$M_3$};    
    \node [box, right=2 of ADD42] (M4) {$M_4$};    
%% machine output    
    \node [right=1 of M1] (f1) {};      
    \node [right=1 of M2] (f2) {};
    \node [right=1 of M3] (f3) {};      
    \node [right=1 of M4] (f4) {};
%% output
    \node [above=0.05 of $(S12)!0.45!(S1)$] (OUT1) {$D_1U_1+D_2U_2$};
    \node [above=0.05 of $(S22)!0.45!(S2)$] (OUT2) {$D_3U_3+D_4U_4$};
    \node [above=0.05 of $(S32)!0.45!(S3)$] (OUT3) {$D_1U_1-D_2U_2$};
    \node [above=0.05 of $(S42)!0.45!(S4)$] (OUT4) {$D_3U_3-D_4U_4$};
%% rectangle
    \draw[thick,dotted]     ($(S1.north west)+(-0.25,0.65)$) rectangle ($(ADD22.south east)+(0.5,-0.65)$);
    \draw[thick,dotted]     ($(S3.north west)+(-0.25,0.65)$) rectangle ($(ADD42.south east)+(0.5,-0.65)$);
    \draw[thick,dotted]     ($(A1p.north west)+(-0.25,0.65)$) rectangle ($(ADD2.south east)+(0.5,-0.65)$);
    \draw[thick,dotted]     ($(A3p.north west)+(-0.25,0.65)$) rectangle ($(ADD4.south east)+(0.5,-0.65)$);
%% coordinate
%\coordinate [right=1 of S1,node distance=.5*\nodedistance] (C1);
%% arrows
    \path[line] (ADD1) edge (S1)
                (A1p) edge (ADD1)
                (A2p) edge (NEG)
                (NEG) edge (ADD2)
                (A2p) edge (ADD1)
                (A1p) edge (ADD2)
                (NEG2) edge (ADD4)
                (A3p) edge (ADD3)
                (A4p) edge (NEG2)
                (A3p) edge (ADD4)
                (A4p) edge (ADD3)
                (S1) edge (ADD12)
                (S2) edge (NEG21)
                (NEG21) edge (ADD22)
                (NEG22) edge (ADD42)
                (S3) edge (ADD32)
                (S4) edge (NEG22)
                (M1) edge (f1)
                (M2) edge (f2)
                (M3) edge (f3)
                (M4) edge (f4)
                (S1) edge (ADD22)
                (S2) edge (ADD12)
                (S3) edge (ADD42)
                (S4) edge (ADD32)
                (ADD2) edge (S32)
                (ADD3) edge (S22)
                (S22) edge (S2)
                (S32) edge (S3)
                (S32) edge (S3)
                (S42) edge (S4)
                (ADD4) edge (S4)
                (ADD12) edge (M1)
                (ADD22) edge (M2)
                (ADD32) edge (M3)
                (ADD42) edge (M4);
\end{tikzpicture}
}
\vspace{-2mm}
\caption{Randomized polar code construction for $N=4$.}
\label{fig:circuit_fig_N_4}
\end{figure*}

%%%%%%%%%%%%%%%%%%
\subsubsection{Randomized Polar Code Construction}
\label{subsec:randpolarcode}
We define the following randomized linear code
\begin{align}
    Z := HDR\,,
\end{align}
where $H$ is the $N\times N$ Hadamard matrix, $D\in \mathbb{R}^{N\times N}$ is a diagonal matrix containing uniform $\pm 1$ Rademacher random variables, and $R\in\mathbb{R}^{N\times s}$ is a $0$--$1$ matrix whose certain rows are set to zeros. The matrix $R$ pads zero entries to frozen data locations, which are determined according to polarized erasure probabilities. The matrix $R$ can be constructed by taking an $s\times s$ dimensional identity matrix and padding all-zero rows at the frozen data indices.

% For the case of a vector data $A \in \mathbb{R}^{s \times 1}$, the encoded data is defined as $\tilde A := Z A$. For the general case of matrix data $A\in\mathbb{R}^{n\times d}$, we partition $A$ into blocks of size $\frac{n}{s} \times d$ as $A=[A_1;\dots;A_s]$. 
% The encoded data is then defined as
%
% \begin{align}
%     \tilde A := Z[\vect{(A_1)}, \dots, \vect{(A_s)}]^T \,,
% \end{align}
% 
For the case of $s=n$, the encoded data is defined as $\tilde A := Z A$. For the general case of matrix data $A\in\mathbb{R}^{n\times d}$, we partition $A$ into blocks of size $\frac{n}{s} \times d$ as $A=[A_1;\dots;A_s]$, where the semicolons indicate that blocks are stacked vertically. 
The encoded data is defined as
\begin{align}
    % \begin{bmatrix} \vect{(\tilde{A}_1)}^T \\ \vdots \\ \vect{(\tilde{A}_N)}^T\end{bmatrix} 
    \tilde{A} &= \mathrm{encode}(A):= [\tilde{A}_1; \dots; \tilde{A}_N] \in \mathbb{R}^{Nn/s \times d} \mbox{, where} \nonumber \\
    &[\vect{(\tilde{A}_1)} , \dots , \vect{(\tilde{A}_N)}]^T := Z[\vect{(A_1)}, \dots, \vect{(A_s)}]^T \,,
\end{align}
where $\vect(\cdot)$ is the columnwise vectorization operator. In the case of the encoded matrix-vector product task, the workers are assigned to compute $\tilde A x$. The encoding function can be equivalently rewritten using the Kronecker product as follows:
\begin{align}
    \tilde{A} &= \mathrm{encode}(A):= (Z \otimes I_{n/s}) A \,,
\end{align}
where $I_{n/s}$ is the $(n/s)\times(n/s)$ dimensional identity matrix.
\subsection{Approximate Computation}

The Hadamard matrix is also used for dimension reduction in randomized approximate algorithms, including the well-known Subsampled Randomized Hadamard Transform (SRHT) \cite{Mahoney11,tropp2011improved}. Until this work, the connection between polar codes and SRHT has not been understood. SRHT is constructed as $S_{H}:=\frac{1}{\sqrt{m}}PHD \in \mathbb{R}^{m\times N}$ where $P \in \mathbb{R}^{m\times N}$ is a $0-1$ row-sampling matrix that picks $m$ rows uniformly at random, $H \in \mathbb{R}^{N\times N}$ is the Hadamard matrix with orthonormal columns, and $D \in \mathbb{R}^{N\times N}$ is a diagonal matrix with diagonal entries sampled i.i.d. from the Rademacher distribution, $D_{ii} = -1 \mbox{ or } 1$ with probability $1/2$. For $m < N$, the sketched data matrix $S_{H}A=\frac{1}{\sqrt{m}}PHDA \in \mathbb{R}^{m\times d}$ can be used as an approximate low dimension version of the data matrix.

We can rewrite the expression for the estimator given in \eqref{eq:estimator}  equivalently as follows:
\begin{align}
    \mathcal{T}_x(\mathcal{S}) &= \vect\left(\frac{1}{|\mathcal{S}|} \sum_{i \in \mathcal{S}} z_i z_i^T [ A_1x, ..., A_s x] \right) \nonumber \\
    &=\vect\left(  \frac{1}{|
    \mathcal{S}|} DH P^T P H D [ A_1x, ..., A_s x] \right) \nonumber \\
    & = 
    \left[ \begin{array}{c}  S_{H}^T S_{H} A_1x\\ \hdots \\  S_{H}^T S_{H}  A_s x \end{array} \right] 
\end{align}
where $P\in \mathbb{R}^{m\times N}$, $m=|\mathcal{S}|$ is a uniform row subsampling matrix that encodes the $|\mathcal{S}|$ workers that finished computation. Therefore, each column of $DH P^T P H D [ A_1x, ..., A_s x]$ is identical to an SRHT sketch applied simultaneously to the $s$ blocks of the desired matrix product $A_1x, ..., A_s x$, where the sketch size equals to $|\mathcal{S}|$.

As the first step of showing that the estimator provides good approximations for the true result, we state Lemma \ref{lem:unbiasedness}, which shows that the approximate results that we obtain using the estimator are unbiased estimates.

\begin{lemma}[Unbiasedness] \label{lem:unbiasedness}
Suppose that the worker job completion times are i.i.d. Then, the estimator $\mathcal{T}_x(\mathcal{S})$ provides unbiased estimates of the true result, i.e., $\Exs[\mathcal{T}_x(\mathcal{S})]=Ax$, where the randomness of the expectation is with respect to the diagonal Rademacher matrix $D$ and  the randomness in the job completion times of the workers.
\end{lemma}
\begin{proof}
Consider the expectation of the estimator $\mathcal{T}_x(\mathcal{S})$ where the randomness is with respect to the randomness of $D$:
\begin{align}
    &\Exs[\mathcal{T}_x(\mathcal{S})] = \Exs \vect\left(\frac{1}{|\mathcal{S}|} DH P^T P H D [ A_1x, ..., A_s x] \right) \nonumber \\
    & = \vect\left(\frac{1}{|\mathcal{S}|} \sum_{i \in \mathcal{S}} \Exs_D\left[ \frac{1}{|\mathcal{S}|} DH \Exs_P[P^T P] H D \right] [ A_1x, ..., A_s x] \right) \nonumber \\
    & = \vect\left([ A_1x, ..., A_s x] \right) =Ax\,.
\end{align}
Note that the third equality follows from the fact $\Exs_P[P^TP]= |\mathcal{S}| I$ due to the i.i.d. distribution of worker job completion times, since the matrix $P$ is a row-sampling matrix whose each row is sampled i.i.d. with replacement. Finally, $H^2=I$ and $\Exs D^2 = I$.
\end{proof}

\subsubsection{Analysis of the anytime estimator}
\label{subsec:proof}
Now we give a proof of the anytime estimation guarantees presented in Theorem \ref{thm:main_result}. 
By the linearity of the map $\mathcal{T}_{x}(S)$ with respect to $x$, we have $\mathcal{T}_{x}(S)-\mathcal{T}_{x^\prime}(S)=\mathcal{T}_{x-x^\prime}(S)$. Thus, we need to prove
\begin{align*}
(1-\epsilon) \|A(x-x^\prime)\|_2^2  &\le \langle \mathcal{T}_{x-x^\prime}(S), A(x-x^\prime) \rangle \\
&\qquad \qquad \le (1+\epsilon) \|A(x-x^\prime)\|_2^2.
\end{align*}
Without loss of generality, we may assume $x^\prime=0$ and it suffices to show
\begin{align*}
(1-\epsilon) \|Ax\|_2^2  \le \langle \mathcal{T}_{x}(S), Ax \rangle \le (1+\epsilon) \|Ax\|_2^2,
\end{align*}
for a fixed vector $x$.

Next, we present the following result on the Johnson-Lindenstrauss (JL) property of the map $HD$ after random erasures, which is a consequence of the analysis of SRHT from the sketching literature \cite{krahmer2011new}:
\begin{theorem}\label{thm:jl_lemma}
Suppose that $P\in \mathbb{R}^{m\times N}$ is an i.i.d. row sampling matrix, where $m\ge O(\log(N)^4)/\epsilon^2$. Then, for any fixed vector $u$, it holds with probability at least $1-\exp(-mc_1)$ that
\begin{align*}
    (1-\epsilon) \|u\|_2^2 \le \|PHDu\|_2^2 \le (1+\epsilon) \|u\|_2^2 \,,
\end{align*}
where $C_1$ is a fixed positive constant.
\end{theorem}
Here, the i.i.d. row-sampling matrix $P$ is achieved by the i.i.d. job completion times of the workers. 
This result implies that the encoding a vector $u$ as $HDu$ followed by random erasures preserves $\ell_2$ norms even though exact recovery may not be possible. Note that for any fixed vector $u$, we have
\begin{align}
    \|PHDu\|_2^2 - \|u\|_2^2 &= u^T (PHD)^T PHD u - \|u\|_2^2 \nonumber \\
    &= u^T S_{H}^TS_H u - \|u\|_2^2 \,. \label{eqn:srht_bound}
\end{align}
It follows from Theorem \ref{thm:jl_lemma} that $|u^T S_{H}^TS_H u - \|u\|_2^2 | \le \epsilon \|u\|_2^2$  with high probability. Next, we apply this bound to bound bilinear terms $u^T S_{H}^TS_H u^\prime$ for any fixed $u$ and $u^\prime$. We combine the last bound with the identity
$$ |u^T Q u^\prime| = \frac{1}{2} | {(u+u^\prime)}^T Q (u+u^\prime)  - u^TQu - {u^\prime}^TQu^\prime|,   $$
which is valid for any symmetric matrix $Q$ by applying $\eqref{eqn:srht_bound}$ three times  to obtain
\begin{align}\label{eq:bilinearbound}
|u^T (S_{H}^TS_H-I) u^\prime| \le \frac{3}{2} \epsilon(\|u\|_2^2+\|u^\prime\|^2_2).
\end{align}
Next, we apply the high-probability bound $|u^T S_{H}^TS_H u - \|u\|_2^2| \le \epsilon \|u\|_2^2$ to each data block by letting $u=A_ix$ for $i=1,...,s$. We obtain
\begin{align*}
&\mathbb{P} \Big[ \exists i \in [\mathcal{S}] ~\textrm{s.t.}~|(A_ix)^T S_{H}^TS_H (A_ix) - \|A_ix\|_2^2| > \epsilon \|A_ix\|_2^2 \Big] \\
&\le s \mathrm{exp}(-mC_1),
\end{align*}
where the inequality follows from the union bound and Theorem \ref{thm:jl_lemma}. Here, $C_1$ is a fixed constant. Finally, adding these inequalities we obtain with high probability the following
\begin{align*}
(1-\epsilon) \sum_{i=1}^s \|A_ix\|_2^2 & \le \sum_{i=1}^s (A_ix)^T S_{H}^TS_H (A_ix) \\
& \qquad \qquad \le (1+\epsilon) \sum_{i=1}^s \|A_ix\|_2^2\,.
\end{align*}
We have $\|Ax\|_2^2 = \sum_{i=1}^s \|A_ix\|_2^2 = 1$ and therefore
\begin{align*}
(1-\epsilon) \|Ax\|_2^2  \le \langle \mathcal{T}_x(S), Ax \rangle \le (1+\epsilon) \|Ax\|_2^2 \,,
\end{align*}
with probability at least $s \mathrm{exp}(-mC_1)$,
which proves the first inequality of Theorem \ref{thm:main_result}. Therefore, the inner product between the estimate $\mathcal{T}_x(\mathcal{S})$ and the vector $Ax$ is close to $\|Ax\|_2^2$ with high probability. 

Next, for fixed arbitrary vectors $\{c_i\}_{i=1}^s$ and $x$ we apply \eqref{eq:bilinearbound} to each data block as follows
\begin{align*}
&\mathbb{P} \Big[ \exists i \in [\mathcal{S}] \textrm{s.t.}|c_i^T S_{H}^TS_HA_ix -c_i^TA_ix| > \frac{3\epsilon}{2} (\|A_ix\|_2^2+\|c_i\|_2^2) \Big] \\
&\le s \, \mathrm{exp}(-mC_1).
\end{align*}
Note that we can scale $\epsilon$ and absorb the constant factors into the constant $C_1$. Adding these inequalities, we obtain with high probability the following
\begin{align}
    |c^T (\mathcal{T}_x(S)-Ax)| \le \epsilon (\|Ax\|_2^2 + \|c\|_2^2)
\end{align}
with probability at least $\mathrm{exp}(-mC^\prime)$ where $c=[c_1^T,...,c_s^T]^T$ and $C^\prime$ is a fixed constant. This proves the second inequality of the theorem. Finally, note that $\mathcal{T}_x(S)$ can be computed in $O(N\log N)$ time using the Fast Hadamard Transform \cite{polar2009arikan}. 

The empirical performance is illustrated for $N=32$ in Figure \ref{fig:averaging_Ax_num_outputs} for a synthetically generated dataset. The error decreases as the number of outputs increases until a decodable set of outputs is detected.

\begin{figure}
  \centering
  \includegraphics[width=0.6\columnwidth]{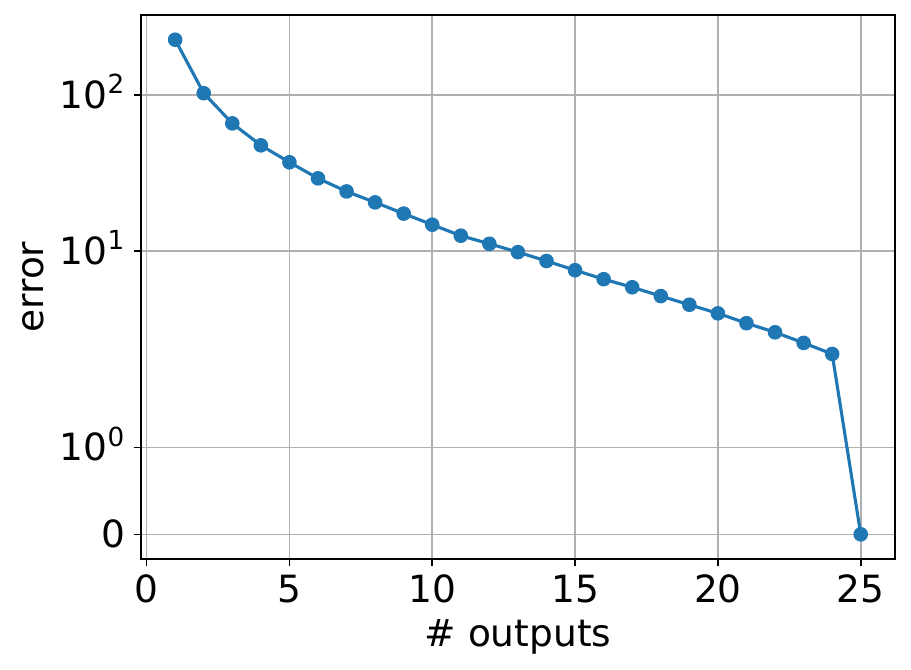}
  \vspace{-3mm}
  \caption{Error as a function of the number of outputs used in computing the estimate. The data matrix $A \in \mathbb{R}^{9600 \times 1000}$ is synthetically generated by sampling from the standard normal distribution. The vector $x \in \mathbb{R}^{1000}$ is also sampled from the standard normal distribution and then scaled by $10^{-3}$. The vertical axis shows the squared $\ell_2$ norm between $Ax$ and the estimate. The number of worker nodes is $N=32$. $24$ inputs are used for data and the remaining $8$ are frozen. For this simulation, the set of available nodes becomes decodable after 25 nodes return their output.}
  \label{fig:averaging_Ax_num_outputs}
\end{figure}

It is worth noting that our approximate recovery result can be viewed as a unification of sketching and coding interpretations of the Hadamard matrix. Specifically, our construction combines the diagonal random Rademacher matrix $D$ from SRHT with the polarization phenomenon from polar codes. As a consequence, we obtain JL embeddings with erasure recovery properties at $O(N\log(N))$ encoding and decoding time. Moreover, since our construction is based on polar codes and sequential decoding, they inherit their capacity achieving properties under random erasures.
\subsection{Exact Recovery via Sequential Decoding}
\label{subsec:exactrec}

In the previous subsection, we have focused on finding an approximate result when $\mathcal{S}$ is not decodable. Now, we discuss the exact recovery of $Ax$ via sequential decoding.

\begin{algorithm}
 \KwIn{The set $\mathcal{S}$ and indices of the frozen inputs}
 
 Initialize $I_{i,j}=$ False for all $i \in [1,N]$, $j \in [1, \log N + 1]$ \;
 
 \For{$i \in \mathcal{S}$} {
    $I_{i, \log N+1}$ = True
 }
 
 Initialize an empty list $y$ \;
 \For{$i \gets 1$ \textbf{to} $N$} {
   $\_, z_{i,1}$ = decodeRecursive($i$, $1$)\;
   \If{node $i$ is not frozen} {
     $y = [y; z_{i,1}]$\;
   }
   \Comment*[r]{forward propagation}
   \vspace{-4mm}
   \If{$i \bmod 2 = 0$ } {
     \For{$j \gets 1$ \textbf{to} $\log{N}+1$} {
       \For{$l \gets 1$ \textbf{to} $i$}{
         compute $z_{l,j}$
         %\Comment*[r]{only need to go up to the level in which $i$ becomes an upper node}
       }
     }
   }
  }
  return $y=Ax$
 \caption{Decoding algorithm}
 \label{decoding_alg}
\end{algorithm}

%%%%%%%%%%%%%%
The decoding algorithm for exact recovery is given in Algorithm \ref{decoding_alg}. The decoder is a sequential algorithm that performs recovery one at a time and in an order. The notation $I_{i,j} \in \{True,False\}$ indicates whether we know the value at node $i$ in level $j$ in the code construction circuit. \textit{Level} indicates the horizontal position while \textit{node} is the vertical position. $z_{i,j}$ is a data structure that holds the value for node $i$ in level $j$.

The decoding algorithm has a subroutine called \textit{decodeRecursive} given in Algorithm \ref{decode_rec}. The main idea behind the decoding algorithm is that it works recursively from right to left (i.e. from output to input) in the code construction. It performs decoding for $2\times 2$ blocks independently and combines the result. The $2\times 2$ blocks are shown as dashed rectangles in Figure \ref{fig:circuit_fig_N_4}. For instance, the top left dashed rectangle has inputs $z_{1,1}, z_{2,1}$ and outputs $z_{1,2},z_{2,2}$. The mapping between the inputs and outputs is given by the Hadamard kernel: $z_{1,2}=z_{1,1}+z_{2,1}$ and $z_{2,2}=z_{1,1}-z_{2,1}$. Hence, when we wish to recover $z_{1,1}$, we sum the outputs and divide by 2. Similarly, to recover $z_{2,1}$, we subtract the second output from the first one and divide by 2. 

In the algorithm, we use the term \textit{pair} to refer to inputs or outputs for a single $2\times 2$ block. In the above example, $z_{1,1}$ and $z_{2,1}$ are a pair and $z_{2,1}$ and $z_{2,2}$ are another pair. The notation $I_{\mbox{pair}(i),j}$ is used to refer to the other node in the pair that node $i$ is in at the $j$'th level. For the above example, $I_{\mbox{pair}(1),1}$ is the same as $I_{2,1}$ since node $2$ is in a pair with node $1$ in the first level. Furthermore, we call the first node in the pair \textit{upper node}. For instance, $z_{1,1}$ and $z_{1,2}$ are upper nodes in their respective pairs. The goal of the decoder algorithm is to compute the values for all the nodes in the first level $z_{i,1}$, $i=1,\dots,N$.

The step of multiplication of inputs $U_i$ by the diagonal matrix $D$ is necessary for the approximate recovery. In the case of exact recovery, the outputs of the decoding algorithm will be the terms $D_iU_ix$. The desired outputs $U_ix$ can be obtained via dividing by $D_i$'s.

\DontPrintSemicolon

\begin{algorithm}
 \KwIn{Node $i \in [1,N]$, level $j \in [1,\log{N}+1]$}
 
  \If{$j = \log{N}+1$} { 
     return ($I_{i,j}$, $z_{i,j}$) if $I_{i,j}=$ True, else ($I_{i,j}$, None) \Comment*[r]{base case 1}
   }
  \If{$I_{i,j} = $ True} {
     return (True, $z_{i,j}$) \Comment*[r]{base case 2}
   }
  $I_{i,j+1} = $ decodeRecursive($i, j+1$)\;
  $I_{\mbox{pair}(i),j+1} = $ decodeRecursive($\mbox{pair}(i), j+1$)\;
 
  \uIf{$i$ is upper node}{
    \If{$I_{i,j+1} \text{ AND } I_{\mbox{pair}(i),j+1} =$ True}{
      compute $z_{i,j}$\;
      return (True, $z_{i,j}$)
    }
  }
  \Else{
    \If{$I_{i,j+1} \text{ OR } I_{\mbox{pair}(i),j+1} =$ True}{
      compute $z_{i,j}$\;
      return (True, $z_{i,j}$)
    }
  }
 return (False, None)
 \caption{decodeRecursive($i$, $j$)}
 \label{decode_rec}
\end{algorithm}

\subsection{Polarization of Computation Times}
\label{sec:comp_times}
In this section, we analyze the time at which the coded computations are decodable. We show that a real valued version of the recursive construction in polar codes enables polarization of the probability density functions towards better or worse computing times, analogous to perfectly noiseless or noisy channels.

\subsubsection{Characterization of Polarizing Kernels} \label{subsec:characterization}

As opposed to classical Polar Codes that operate in finite fields, our constructions have more freedom in their design space over the real numbers. In particular, the polarization kernels can be chosen arbitrarily as long as polarization takes place. In this section, we provide a characterization of matrices that enable polarization and exact recovery. We first give the definition of a \textit{polarizing kernel}, and then state Lemma \ref{polarizing_kernel_lemmma} that characterizes the conditions for a polarizing kernel.
\begin{definition}[Polarizing $2\times 2$ kernel over reals] \label{polarizing_kernel_def}
Let $f$ be a function satisfying the linearity property $f(au_1+bu_2) = af(u_1)+bf(u_2)$ where $a,b \in \mathbb{R}$ and assume that there is an algorithm to compute $f$ that takes a certain amount of time to run with its run time distributed randomly. Let $K$ denote a $2\times 2$ kernel and $\bigl[\begin{smallmatrix} v_1 \\ v_2 \end{smallmatrix}\bigr] = K \times \bigl[\begin{smallmatrix} u_1 \\ u_2 \end{smallmatrix}\bigr]$. Assume that we input $v_1$ and $v_2$ to two i.i.d. instances of the same algorithm for $f$. Further, let $T_1$, $T_2$ be random variables denoting the run times for computing $f(v_1)$, $f(v_2)$, respectively. We are interested in computing $f(u_1)$, $f(u_2)$ in this order. If the time required to compute $f(u_1)$ is $\max(T_1,T_2)$ and the time required to compute $f(u_2)$ given the value of $f(u_1)$ is $\min(T_1,T_2)$, then we say $K$ is a polarizing kernel. Note that this definition exclusively considers the earliest time at which $f(u_i)$ can be computed, without accounting for the decoding time involved.
\end{definition}

\begin{lemma} \label{polarizing_kernel_lemmma}
A kernel $K \in \mathbb{R}^{2 \times 2}$ is a polarizing kernel if and only if the following conditions are both satisfied: 1) Both elements in the second column of $K$ are non-zero, 2) $K$ is invertible.
\end{lemma}

Theorem \ref{thm_decoder_kernel} builds on Lemma \ref{polarizing_kernel_lemmma} to identify the polarizing kernels that require the least amount of computations for encoding. 

\begin{theorem} \label{thm_decoder_kernel}
 Of all possible $2\times 2$ polarizing kernels, the kernels $F_2=\bigl[\begin{smallmatrix}1 & 1\\ 0 & 1\end{smallmatrix}\bigr]$ and $F_2^\prime=\bigl[\begin{smallmatrix}0 & 1\\ 1 & 1\end{smallmatrix}\bigr]$ result in the fewest number of computations for encoding real-valued data.
\end{theorem}

The proofs of Lemma \ref{polarizing_kernel_lemmma} and Theorem \ref{thm_decoder_kernel} are in the Appendix. Note that the results so far apply to only size $2$ kernels. Next, we extend the results to kernels of bigger size.

\begin{definition}[Polarizing $p\times p$ kernel over reals] \label{def:polarizing_kernel_def_size_p}
This definition extends Definition \ref{polarizing_kernel_def} for polarizing kernels to arbitrary size kernels. Let $\pi_i$ denote the index of the node $i$ when the run times of nodes are sorted in decreasing order, $T_{\pi_1} \geq T_{\pi_2} \geq \dots \geq T_{\pi_p}$. If the time required to compute $f(u_i)$ is equal to $T_{\pi_i}$ for all $i=1,\dots,p$, then we say it is a polarizing kernel.
\end{definition}

%%%%%%%%%
\noindent \textbf{(i) Kernel Size 3:}
Consider the $3\times 3$ kernel $K = \Bigl[\begin{smallmatrix} a & b & c \\ d & e & f \\ g & h & i \end{smallmatrix}\Bigr]$. Then, $K$ is a polarizing kernel if and only if
\begin{enumerate}
    \item $K$ is invertible over the reals
    \item All of the matrices $\bigl[\begin{smallmatrix} b & c \\ e & f \end{smallmatrix}\bigr]$, $\bigl[\begin{smallmatrix} b & c \\ h & i \end{smallmatrix}\bigr]$, $\bigl[\begin{smallmatrix} e & f \\ h & i \end{smallmatrix}\bigr]$ are invertible over the reals
    \item $c,f,i$ are all non-zero.
\end{enumerate}
We note that the following kernel does not require any multiplications and only requires additions and subtractions and also is a polarizing kernel: $K = \Bigl[\begin{smallmatrix} 1 & 1 & 1 \\ 0 & -1 & 1 \\ 0 & 0 & 1 \end{smallmatrix}\Bigr]$.

%%%%%%%%%
\noindent \textbf{(ii) Arbitrary Kernel Size:}
Let $K$ be a $p\times p$ kernel. It is a polarizing kernel if and only if it satisfies the following conditions:
\begin{enumerate}
    \item $K$ is invertible
    \item After removing the first column of $K$, every $p-1$ rows of the remaining matrix is a matrix invertible over the reals
    \item After removing the first and second columns of $K$, every $p-2$ rows of the remaining matrix is a matrix invertible over the reals
    \item[\vdots] 
    \item[$p-1$)] After removing the first $(p-1)$ columns of $K$, every scalar in the remaining matrix is non-zero (i.e. $1\times 1$ invertible matrix).
\end{enumerate}
Note that for $p=4$, the following upper triangular matrix is a polarizing kernel:\begin{center} $K = \biggl[\begin{smallmatrix} 1 & 1 & 1 & 1 \\ 0 & 1 & 2 & 3 \\ 0 & 0 & 1 & 4 \\ 0 & 0 & 0 & 1\end{smallmatrix} \biggr].$\end{center}

%%%%%%%%%%%
\subsubsection{Recursive Polarization} \label{subsec:polarization}

As in polar codes, we now consider the recursive application of any $2 \times 2$ polarizing kernel via the Kronecker power construction $K_{2N} = K \otimes K_N$ and $K_2=K$. For instance, one can take $K=F_2=\bigl[\begin{smallmatrix}1 & 1\\ 0 & 1\end{smallmatrix}\bigr]$ to obtain the real valued version of polar codes. This operation is depicted for $F_4$ in Figure \ref{fig:circuit_fig_N_4}. Note that the input is permuted to bit reversed order in Figure \ref{fig:circuit_fig_N_4}. We refer the reader to \cite{polar2009arikan} for a detailed description of the bit reversal process. In the case of the Hadamard kernel $H_2$, the recursive construction coincides with the Hadamard transformation. We now analyze the run-times in the recursive construction by illustrating the $4\times 4$ construction.
\begin{lemma}
\label{lem:recursive}
Suppose that the kernel $K$ is polarizing as given in Definition \ref{polarizing_kernel_def} and Lemma \ref{polarizing_kernel_lemmma}, and let $T_1,\dots,T_4$ be random i.i.d. run-time random variables. Then the sequential decoding procedure described in Algorithm 2 computes $f(u_1),f(u_2),f(u_3),f(u_4)$ respectively in time
\begin{align*}
%T^{(1)} \rightarrow 
T^{(1)}&= \max\big(\max( T_{1}, T_{2} ), \max( T_{3}, T_{4} )\big) \\
%T^{(2)} \rightarrow 
T^{(2)}&= \min\big( \max( T_{1}, T_{2}  ), \max( T_{3}, T_{4} )\big) \\
T^{(3)}&= \max\big( \min( T_{1}, T_{2}  ), \min( T_{3}, T_{4} )\big) \\
T^{(4)}&= \min\big( \min( T_{1}, T_{2}  ), \min( T_{3}, T_{4} )\big).
\end{align*}
\end{lemma}
\begin{proof}
Note that the recursive construction (e.g. in Figure \ref{fig:circuit_fig_N_4}) combines two independent random run-times and transforms to $\max(T_1,T_2)$ and $\min(T_1,T_2)$. Recursively applying this transformation proves the statement.
\end{proof}

Another way to show the run-time transformations for $N=4$ would be to consider that to recover the first input, we need to know both outputs of the top left dashed rectangle in Figure \ref{fig:circuit_fig_N_4}. That requires knowing all 4 outputs; hence $T^{(1)}$ is equal to maximum of all $T_i$'s. Next, given the first input, to recover the second one, we need only one of the outputs of the top left dashed rectangle. Therefore, we obtain that $T^{(2)}= \min\big( \max( T_{1}, T_{2}  ), \max( T_{3}, T_{4} )\big)$. The idea is similar for the  recovery of the third and fourth inputs.

The result of Lemma \ref{lem:recursive} naturally extends to constructions of arbitrary sizes where the corresponding run times are alternating $\min$ and  $\max$ expressions. Consequently, one can freeze certain variables to obtain a faster overall run-time. For example, freezing the first input by setting it to a fixed value (e.g., zero), the decoder leverages this knowledge to eliminate the run-time $T^{(1)}= \max\big(\max( T_{1}, T_{2} ), \max( T_{3}, T_{4} )\big)$.

\noindent \textbf{Example 1.} Let us illustrate Lemma \ref{lem:recursive} for the case of uniformly random run time distributions with closed-form formulas. Suppose that $T_1,T_2,T_3,T_4$ are i.i.d. and uniform in the interval $[0,1]\subseteq\mathbb{R}$. Then a straightforward calculation shows that the probability density functions of the run times $T^{(1)},T^{(2)},T^{(3)},T^{(4)}$ are given by
\begin{align*}
    p_{T^{(1)}}(t) &= 4t^3\\
    p_{T^{(2)}}(t) &= 4t(1-t)(1+t)\\
    p_{T^{(3)}}(t) &= 4t(1-t)(2-t)\\
    p_{T^{(4)}}(t) &= 4(1-t)^3
\end{align*}
for $t \in [0,1]$. It can be seen that the probability density functions are degree $N-1$ polynomials for the general size $N$ construction with the uniform distribution.\\
%%%%%%%%%%%%%%%%%%%%%%%%%%%%%
\noindent \textbf{Polarization for Kernel Size 3:}
Let us denote the inputs by $x_1,x_2,x_3$ and the corresponding channel outputs by $y_1,y_2,y_3$. Next, note that
$x_1$ can be recovered when all three outputs $y_1,y_2,y_3$ are known. Hence the run time for $x_1$ is equal to $\max(T_1, T_2, T_3)$. Given $x_1$, decoding $x_2$ requires at least two of $y_1,y_2,y_3$. The run-time for $x_2$ is equal to $\median(T_1,T_2,T_3)$. Given $x_1, x_2$, the decoding of $x_3$ will take $\min(T_1,T_2,T_3)$. 

Note that the run times for a construction of size $N=9$ can be obtained using the same method that we previously described for computing the run times for $N=4$ with kernel size $2$ in Lemma \ref{lem:recursive}. In particular, the run times for $N=9$ will be as follows: Decoding $x_1$ will take time $\max(\max(T_1,T_2,T_3), \max(T_4,T_5,T_6), \max(T_7,T_8,T_9)) = \max(T_1,\dots,T_9)$. Decoding $x_2$ will take time $\median(\max(T_1,T_2,T_3), \max(T_4,T_5,T_6), \max(T_7,T_8,T_9))$. Decoding the last input $x_9$ will take time $\min(T_1,\dots,T_9)$.

%%%%%%%%%%%%%%%%%%%%%%%%%%%%%
\noindent\textbf{Polarization for Arbitrary Kernel Sizes:}
Let $p \geq 2$ denote the kernel size. Let $\pi_i$ denote the index of the node $i$ when the run times of nodes are sorted in decreasing order, $T_{\pi_1} \geq T_{\pi_2} \geq \dots \geq T_{\pi_p}$. Then, for kernels of arbitrary size $p$, the run-time required for decoding $x_i$ is equal to $T_{\pi_i}$. It is easy to observe that polarization for kernel sizes 2 and 3 are a special case of this result.

\section{Numerical Results} \label{sec:numerical_results}

In this section, we present numerical results to verify our theoretical claims and test the performance of the proposed methods in various tasks.

\subsection{Polarization of Computation Times}
Figure \ref{polarized_cdfs} is a visualization of the polarization of computation times. Plot (a) shows the empirical cumulative distribution function (CDF) of computation times for serverless functions in AWS Lambda. This empirical CDF has been obtained by running the same Python script in $500$ serverless functions in parallel in AWS Lambda. Furthermore, observe that plot (a) shows that there are worker nodes that finish their computations much later in roughly $t=120$ seconds as opposed to many worker nodes that finish before $t=20$ seconds.

Plot (b) of Figure \ref{polarized_cdfs} shows the CDF for the transformed computation times for $2$ workers. This has been simulated by assuming there are $2$ worker nodes and they have i.i.d. computation times with CDF shown in plot (a). Similarly, plot (c) and (d) show the CDFs for $16$ and $64$ worker nodes. This process shows that we transform the computation times into better and worse computation times, i.e. polarization of the computation times. Freezing the inputs with worse computation times leads to a straggler-resilient computation mechanism since this is the same as picking only the transformed nodes with better computation times to perform the actual computation.

\begin{figure}%[htb]
\begin{minipage}{0.4\columnwidth}
  \centering
  \centerline{\includegraphics[width=\columnwidth]{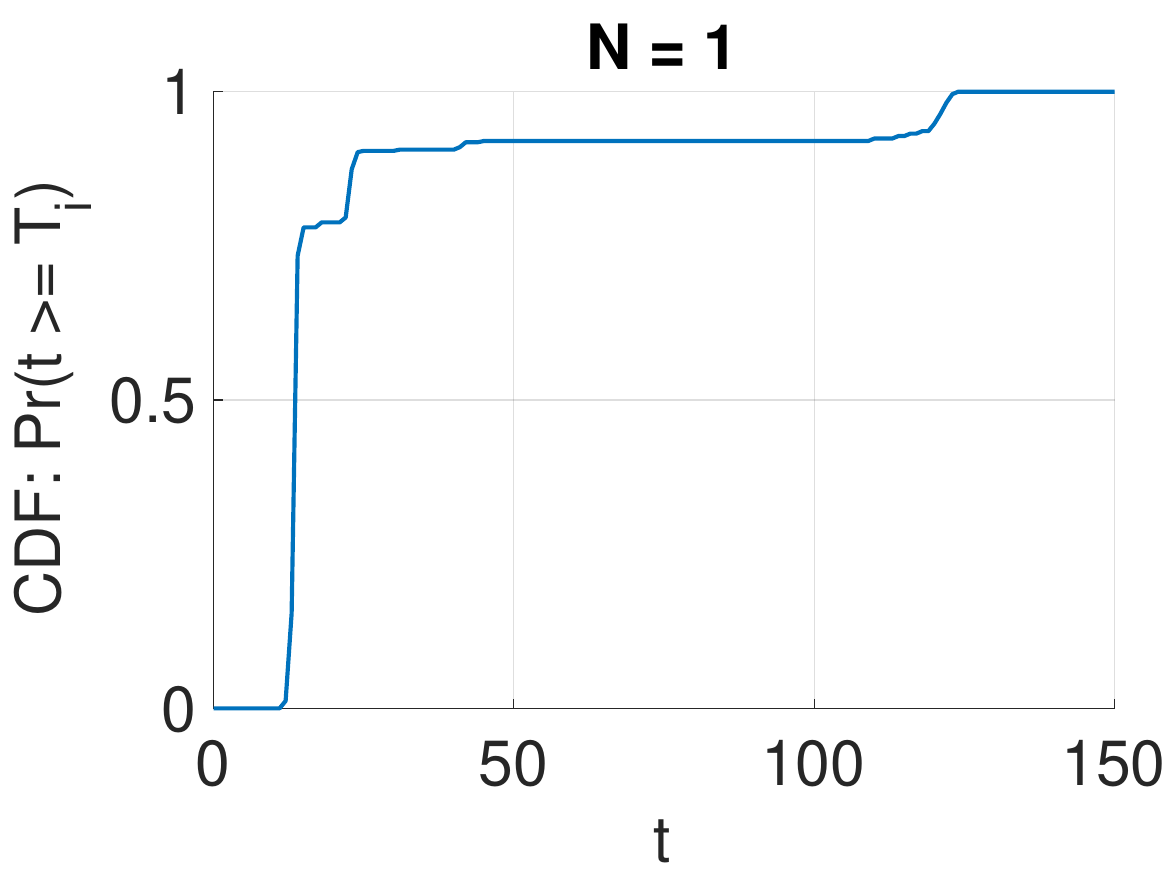}}
  \vspace{-2mm}
  \centerline{\footnotesize (a) $N=1$}
\end{minipage}
\hfill
\begin{minipage}{0.4\columnwidth}
  \centering
  \centerline{\includegraphics[width=\columnwidth]{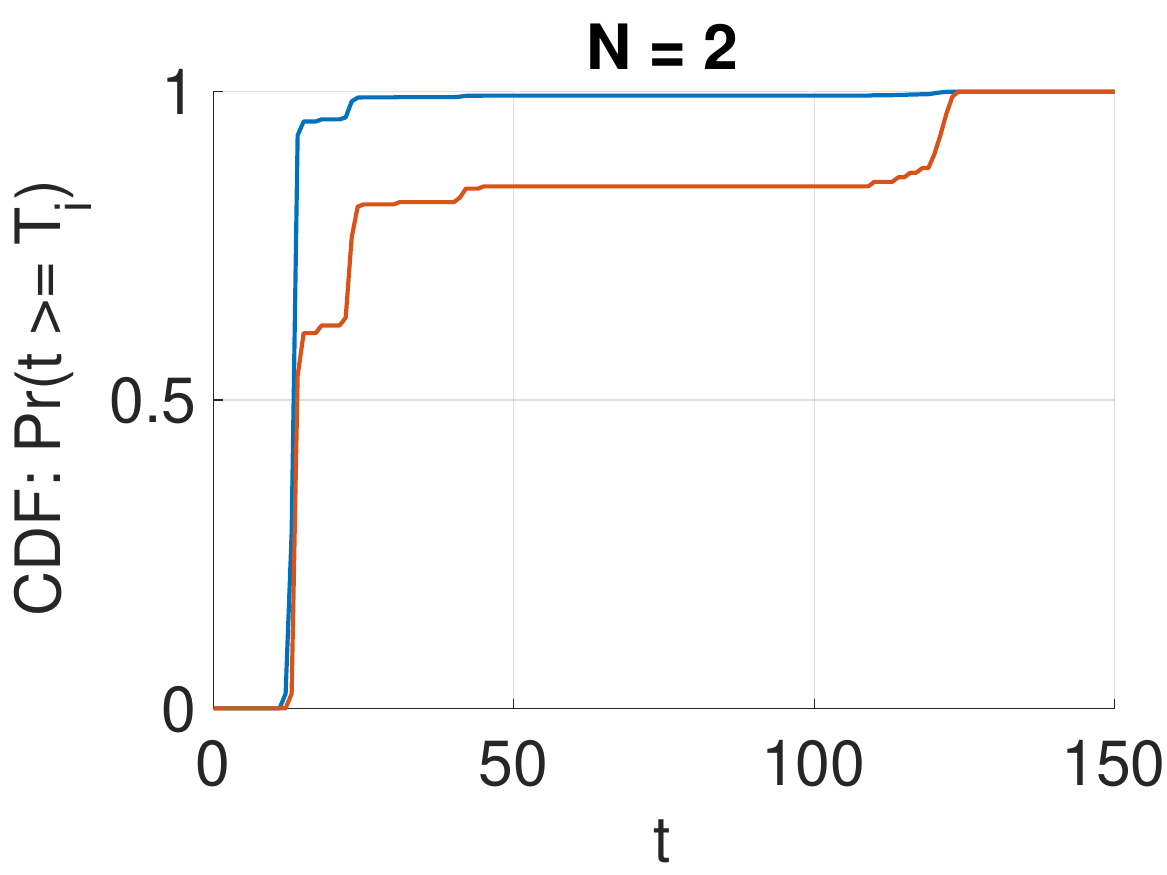}}
  \vspace{-2mm}
  \centerline{\footnotesize (b) $N=2$}
\end{minipage}
% \hfill
\begin{minipage}{0.4\columnwidth}
  \centering
  \centerline{\includegraphics[width=\columnwidth]{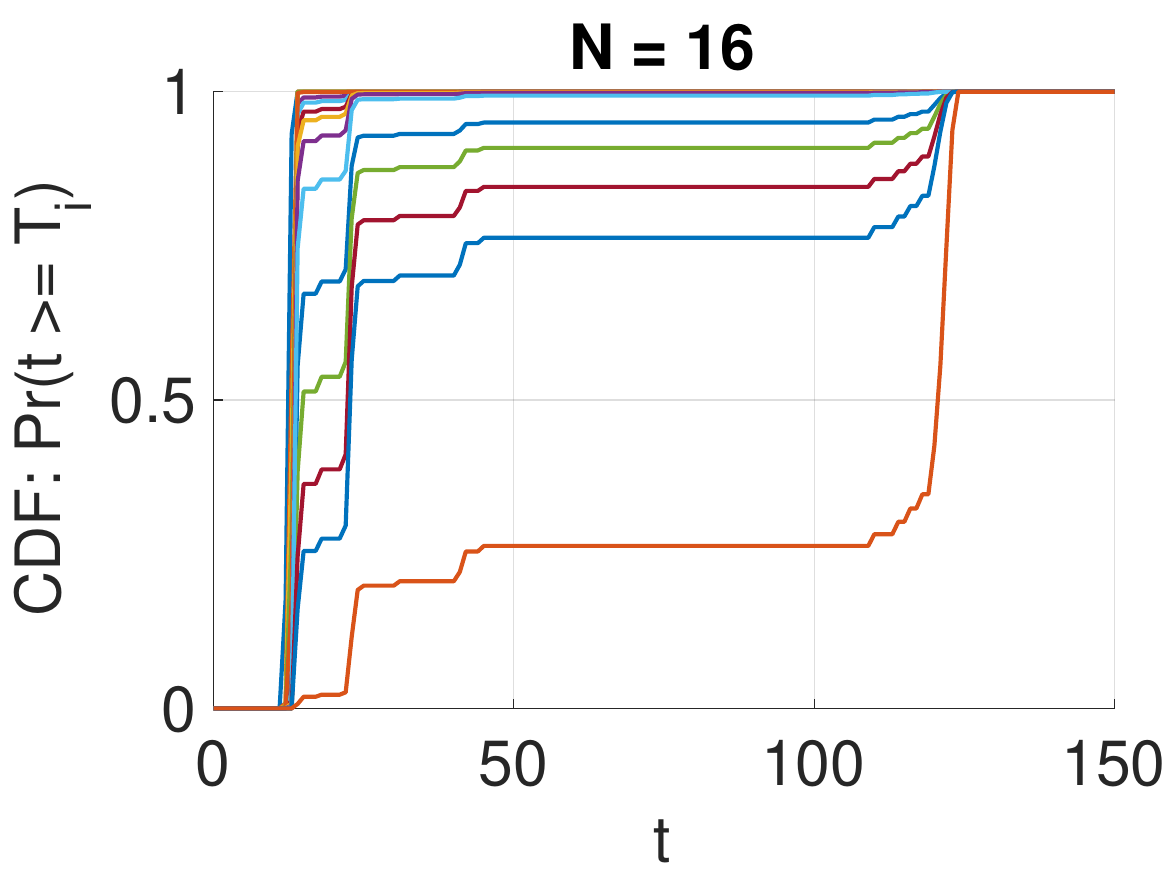}}
  \vspace{-2mm}
  \centerline{\footnotesize (c) $N=16$}
\end{minipage}
\hfill
\begin{minipage}{0.4\columnwidth}
  \centering
  \centerline{\includegraphics[width=\columnwidth]{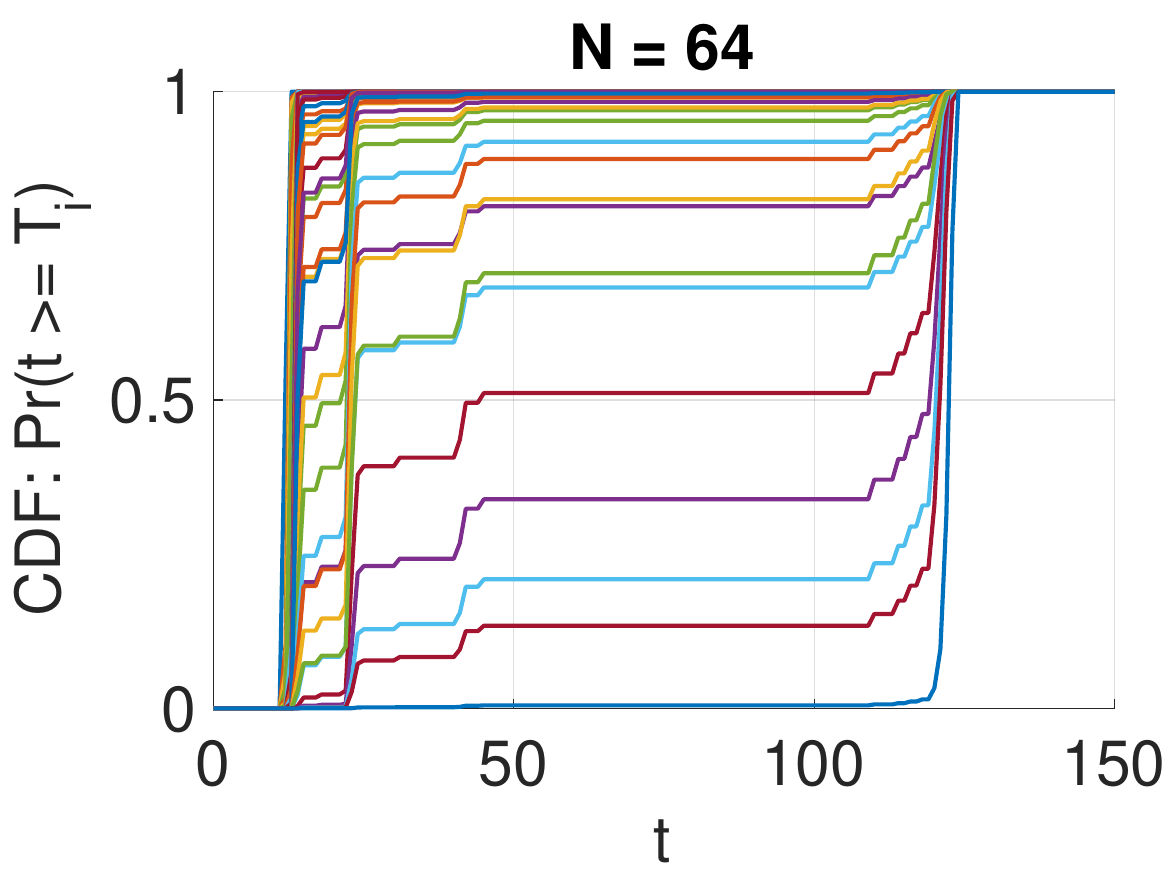}}
  \vspace{-2mm}
  \centerline{\footnotesize (d) $N=64$}
\end{minipage}
% \vspace{-3mm}
\caption{Empirical CDFs for various construction sizes show the polarization of run times.}
\label{polarized_cdfs}
\end{figure}

\subsection{Gradient Descent for the Least Squares Problem}
Polar coded distributed computation method can be used in any algorithm that requires matrix multiplication. Consider the gradient descent algorithm being applied to solve a linear least squares problem $\text{minimize}_x \, \|Ax-y\|_2^2$,
where $A$ is a large-scale data matrix. The variable $x$ is small enough to fit in the memory of a worker node. The update rule for gradient descent is as follows:
\begin{align} \label{eq_update_rule}
    x_{t+1} = x_t - \mu (A^TAx_t - A^Ty),
\end{align}
where the subscript $t$ in $x_t$ denotes the iteration number. Note that we can compute $A^Ty$ only once and but need to compute $A^TAx_t$ multiple times across iterations.

One possible scenario is to encode both $A$ and $A^T$ separately and then use the coded matrix-vector multiplication method twice every iteration; once for $Ax_t$ using the coding on $A$ and once for $A^T(Ax_t)$ using the coding on $A^T$. Another choice is to compute $A^TA$ offline and encode directly the product $A^TA$. In this case, we use the coded matrix-vector multiplication method only once per iteration.

Figure \ref{plot_GD_8} compares the uncoded and the polar coded distributed computation methods along with different values for $N$ and erasure probabilities $\epsilon$. In this experiment we have pre-computed and encoded the matrix product $A^TA$.  Then, in each iteration of the gradient descent, the central node decodes the downloaded outputs, updates $x_t$, sends the updated $x_t$ to AWS S3 and initializes the computation $A^TAx_t$. In the case of uncoded computation, we simply divide the multiplication task among $N (1-\epsilon)$ serverless functions, and whenever all of the $N (1-\epsilon)$ functions finish their computations, the outputs are downloaded to the central node, and there is no decoding. Then, the central node computes and sends the updated $x_t$, and initializes the next iteration. The data matrix has dimensions $A \in \mathbb{R}^{20000\times 4800}$, the variable is $x \in \mathbb{R}^{4800\times 1000}$, and the output is $y \in \mathbb{R}^{20000\times 1000}$. We have randomly generated the data used in this experiment.

We note that in a given iteration, while computation with polar coding with rate $(1-\epsilon)$ requires waiting for the first decodable set of outputs out of $N$ outputs, uncoded computation waits for all $N (1-\epsilon)$ nodes to finish computation. Using $\epsilon$ as a tuning parameter for redundancy, we achieve different convergence times.

\subsection{ImageNet: Large-Scale Experiment}
Figure \ref{fig:imagenet} shows the cost ($\|Ax-y\|_2^2$) against wall-clock time when we solve the least squares problem where the data matrix $A$ consists of the first $128$ classes of the ImageNet dataset \cite{krizhevsky2012imagenet}. This experiment aims to demonstrate that gradient descent with coded matrix multiplication can be used to speed up fine-tuning of pre-trained machine learning models. 

Each data sample of the ImageNet dataset is an RGB picture of (rescaled) dimensions $256\times 256 \times 3$. Figure \ref{fig:imagenet} compares the computation speeds of the naive approach with no coding (orange) and partial coded construction (blue). The circles show the beginning of each iteration. For the partial coded construction, the construction size is $2$, which is equivalent to the repetition coding. Both methods have been run for $30$ iterations with $5$ iterations per serverless function lifetime. In other words, each serverless function has been reused for $5$ iterations. Figure \ref{fig:imagenet} demonstrates that coding, even for a small construction size of $2$ helps speed up the computation. Larger construction sizes are expected to reduce the computation times further at the expense of increased encoding times.

\begin{figure}
\begin{minipage}{\columnwidth}
  \centering
  \centerline{\includegraphics[width=0.58\textwidth]{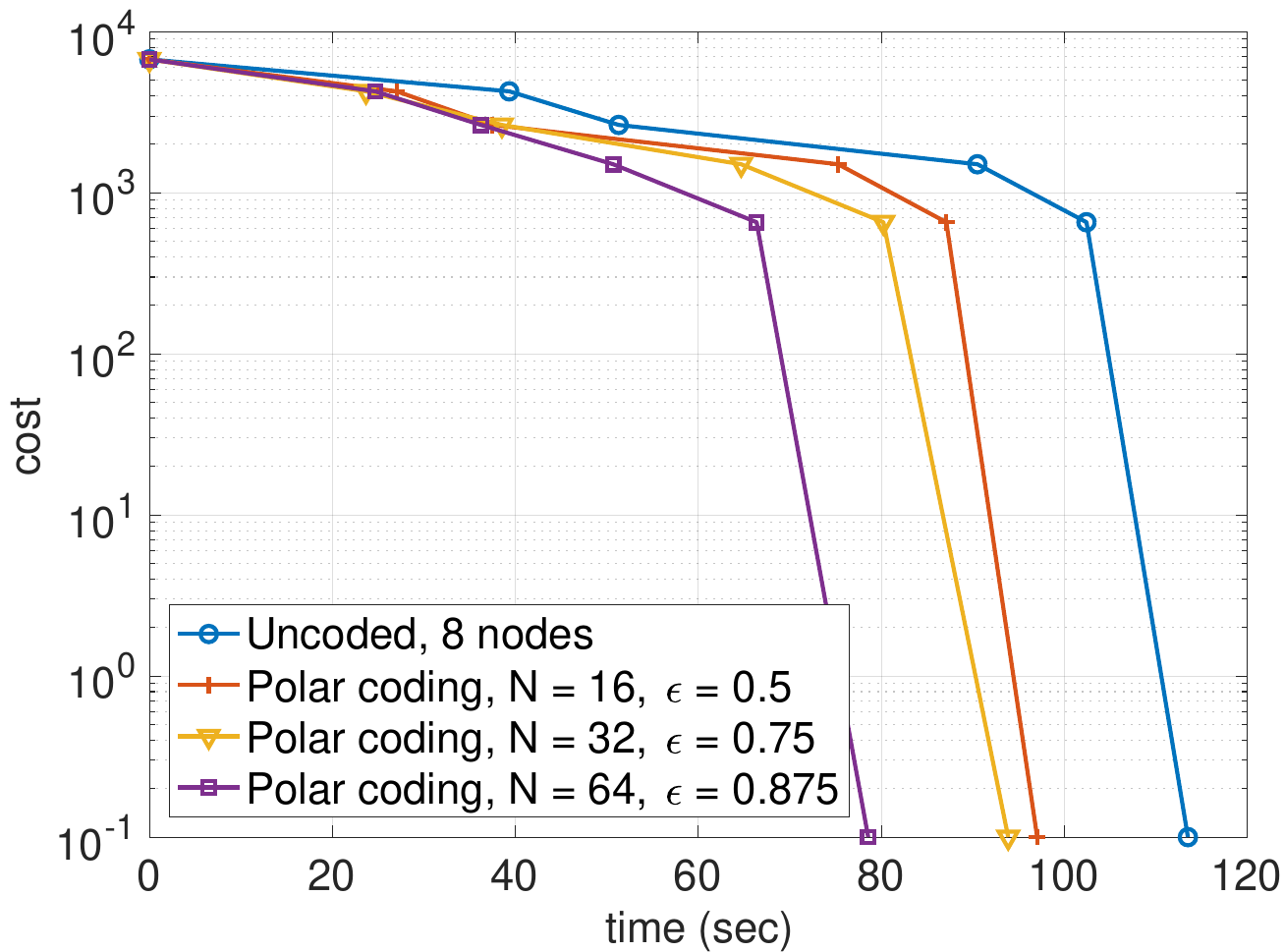}}
  \vspace{-3mm}
  \caption{Cost ($\|Ax-y\|_2^2$) against wall-clock time for the gradient descent example.}\medskip
  \label{plot_GD_8}
\end{minipage}
\hfill
\begin{minipage}{\columnwidth}
  \centering
  \centerline{\includegraphics[width=0.65\textwidth]{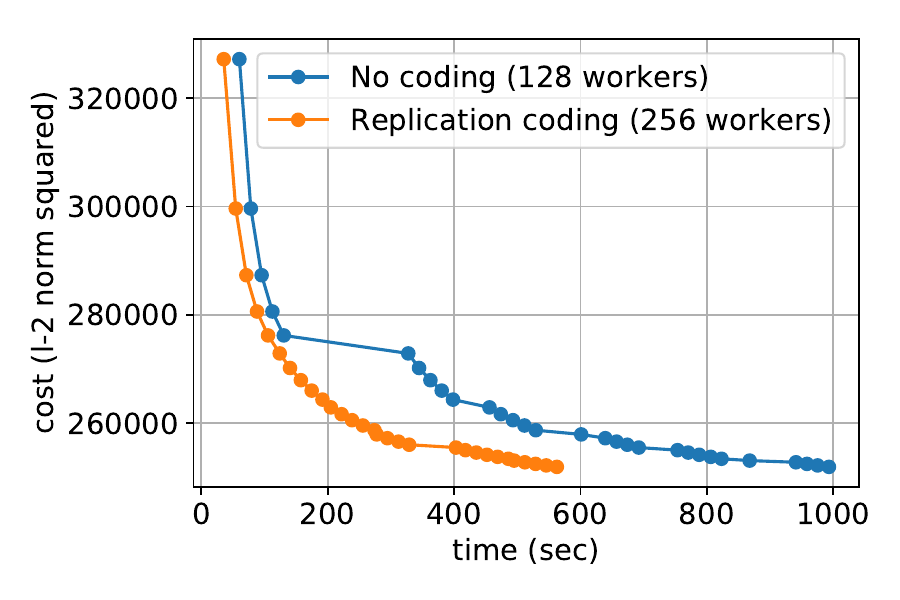}}
  \vspace{-3mm}
  \caption{Large-scale gradient descent example on ImageNet.}\medskip
  \label{fig:imagenet}
\end{minipage}
\end{figure}

We observe that the fifth iteration for the uncoded case takes longer than $100$ seconds and there are other iterations that take much longer than the rest of the iterations. This is expected when we do not have redundancy since even if there is only one straggling worker node, the central node waits for the straggling node to return its output before starting the next iteration of the gradient descent algorithm. Figure \ref{fig:imagenet} verifies that this effect is mitigated with replication coding.
Furthermore, we note that the coded version achieves roughly a 50\% reduction in the computation time while using twice many AWS Lambda functions. Thus, the price of the overall computation stays the same since pricing is calculated based on the duration of time that the functions take executing. This experiment verifies the effectiveness of our approach as it leads to faster computing while keeping the overall price the same.

\subsection{Extension to Optimizing Nonlinear Black-Box Functions}
In this subsection, we present numerical results for the  application of our randomized polar codes to black-box optimization problems. This is an extension of our approach to computing linear functions by considering a linear approximation to the gradient for nonlinear objective functions. The details are given in the appendix. Our black-box optimization strategy only accesses the objective through function evaluations. This makes our method applicable to any arbitrary objective function.  We consider the optimization of the nonlinear objective $f(\theta) = \|A\theta - b\|_1$ as a test case for the proposed black-box optimization method.

\begin{figure}
% \begin{wrapfigure}{r}{0.4\textwidth}
  % \centering
  \begin{center}
  \includegraphics[width=0.55\columnwidth]{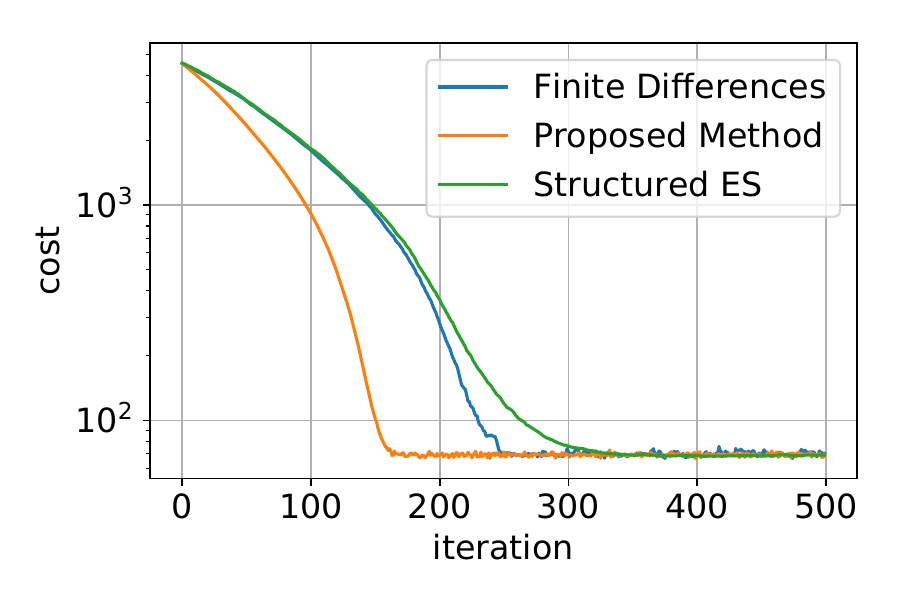}
  \end{center}
  \vspace{-5mm}
  \caption{Black-box optimization of the function $f(\theta)=\|A\theta-b\|_1$ where $A \in \mathbb{R}^{200\times 32}, b \in \mathbb{R}^{200}$.}
  \label{l1_norm_example_comparison}
% \end{wrapfigure}
\end{figure}

Figure \ref{l1_norm_example_comparison} shows the cost as a function of iterations when we use gradient descent algorithm with gradient estimates obtained by the finite differences method, the proposed coded black-box optimization method, and the structured evolution strategies method (these methods are described in the appendix). 
For fairness in comparing these methods, we used straggler-resilient versions of the finite differences method and the structured evolution strategies method in obtaining these results. We wait for the first arriving $16$ worker outputs out of the total $32$ outputs to make a gradient update for the finite differences method. For the structured evolution strategies method, the perturbation directions are generated from the rows of the matrix product $HD$ which is a $32\times 32$ matrix, so we wait for the first arriving $16$ outputs out of the $32$. Finally, we implement the proposed method by utilizing a rate of $\frac{1}{2}$ with a total of $64$ workers and wait for the first decodable set of outputs out of these $64$ outputs. Figure \ref{l1_norm_example_comparison} illustrates that having all the entries of the gradient estimate through decoding leads to faster convergence than having only a half of the entries of the gradient estimate. It also shows that the proposed method results in faster convergence compared to the structured evolution strategies method.

\subsection{Encoding and Decoding Speed Comparison}
Figure \ref{coding_time_comparison_2} shows the time that encoding and decoding algorithms take as a function of the number of nodes $N$ for Reed-Solomon codes and our approach for exact recovery using polar codes. For Reed-Solomon codes, we have implemented two separate approaches for encoding and decoding. The first is the naive approach where encoding is done using matrix multiplication ($O(N^2)$) and decoding is done by solving a linear system ($O(N^3)$), hence the naive encoder and decoder can support full-precision data. The second approach is the fast implementation for both encoding and decoding (of complexities $O(N\log N)$ and $O(N\log^2 N)$, respectively). The fast implementation is based on Fermat Number Transform (FNT), hence only supports finite field data. In obtaining the plots in Figure \ref{coding_time_comparison_2}, we used $0.5$ as the rate and performed the computation $Ax$ where $A$ is $(100N \times 5000)$-dimensional and $x$ is $(5000\times 1000)$-dimensional. The curve in Figure \ref{coding_time_comparison_2}(b) with cross markers and dashed lines, labeled as 'high error', indicates that the error due to the decoder is unacceptably high. This happens since the linear system that we need to solve for recovery is ill-conditioned.

\begin{figure}%[htb]
\centering
\begin{minipage}[b]{0.47\columnwidth}
  \centering
\centerline{\includegraphics[width=\textwidth]{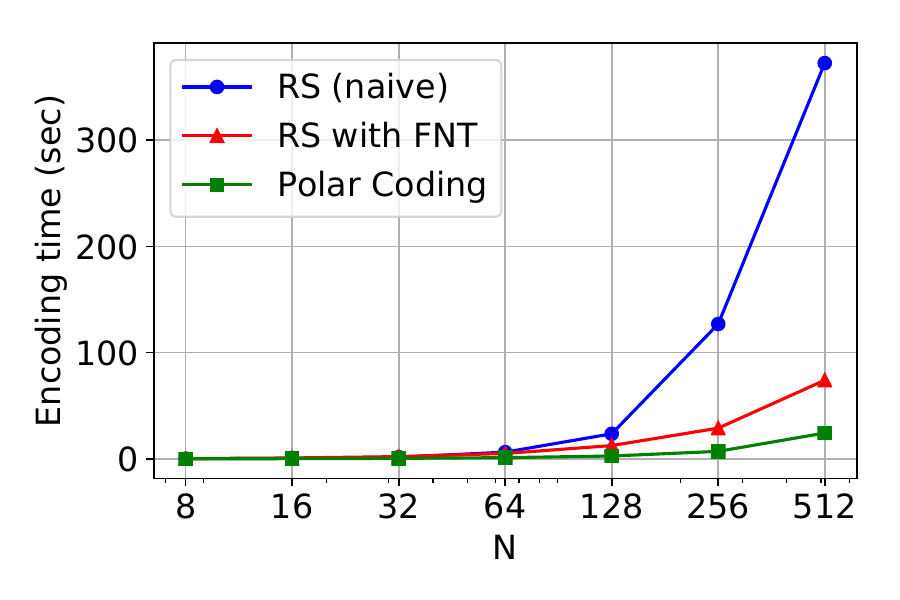}}
  \vspace{-3mm}
  \centerline{\footnotesize (a) Encoding}\medskip
\end{minipage}
\hfill
% \hspace{5mm}
\begin{minipage}[b]{0.47\columnwidth}
  \centering
\centerline{\includegraphics[width=\textwidth]{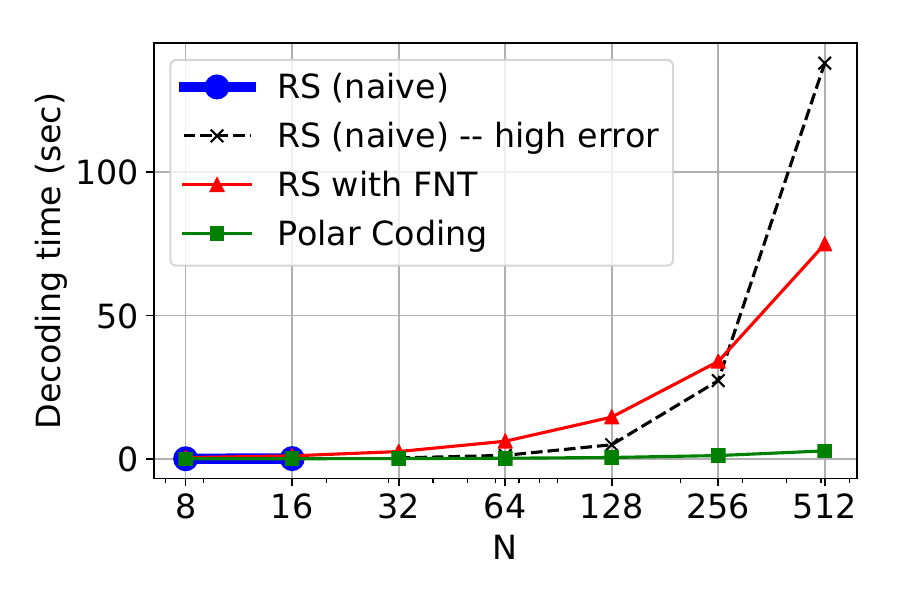}}
  \vspace{-3mm}
  \centerline{\footnotesize (b) Decoding}\medskip
\end{minipage}
\vspace{-5mm}
\caption{Comparison of encoding and decoding speeds for RS and polar codes.}
\label{coding_time_comparison_2}
\end{figure}

Figure \ref{coding_time_comparison_2} illustrates that our approach for exact recovery takes considerably less time for encoding and decoding compared to Reed-Solomon codes. This is because of the constants hidden in the complexities of fast decoders for Reed-Solomon decoders which is not the case for polar codes. We note that it might be more advantageous to use RS codes for small $N$ values because they have the MDS properties and can encode and decode fast enough for small $N$. However, in serverless computing where each function has limited resources and hence using large $N$ values is usually the case, we need faster encoding and decoding algorithms. Considering the comparison in Figure \ref{coding_time_comparison_2}, our approach using polar codes is more suitable for serverless computing where $N$ is large.

%%%%%%%%%%%%%%%%
\subsection{Empirical Distribution of Decodability Times}
We refer to the time instance where the available outputs become decodable for the first time as \textit{decodability time}. Figure \ref{decod_time_histograms} shows the histograms of the decodability time for different values of $N$ for polar, LT, and MDS codes, respectively. These histograms were obtained by sampling i.i.d. worker run times with replacement from the input distribution whose CDF is plotted in Figure \ref{polarized_cdfs}(a) and by repeating this $1000$ times. Further, $\epsilon=0.375$ was used as the erasure probability. We observe that as $N$ increases, the distributions converge to the dirac delta function, showing that for large $N$ values, the decodability time becomes deterministic.

Plots in Figure \ref{decod_time_histograms}(d,e,f) are the decodability time histograms for LT codes with peeling decoder. The degree distribution is the robust soliton distribution as suggested in \cite{mallick2018rateless}. We see that polar codes achieve better decodability times than LT codes. Plots in Figure \ref{decod_time_histograms}(g,h,i) on the other hand show that MDS codes perform better than polar codes in terms of decodability time, which is expected. When considering this result, one should keep in mind that for large $N$, MDS codes take much longer times to encode and decode compared to polar codes as we discussed previously. In addition, we see that for large $N$ values, the gap between the decodability time performances closes.

\begin{figure}%[ht!]
\begin{minipage}[b]{0.28\columnwidth}
  \centering
  \centerline{\includegraphics[width=\textwidth]{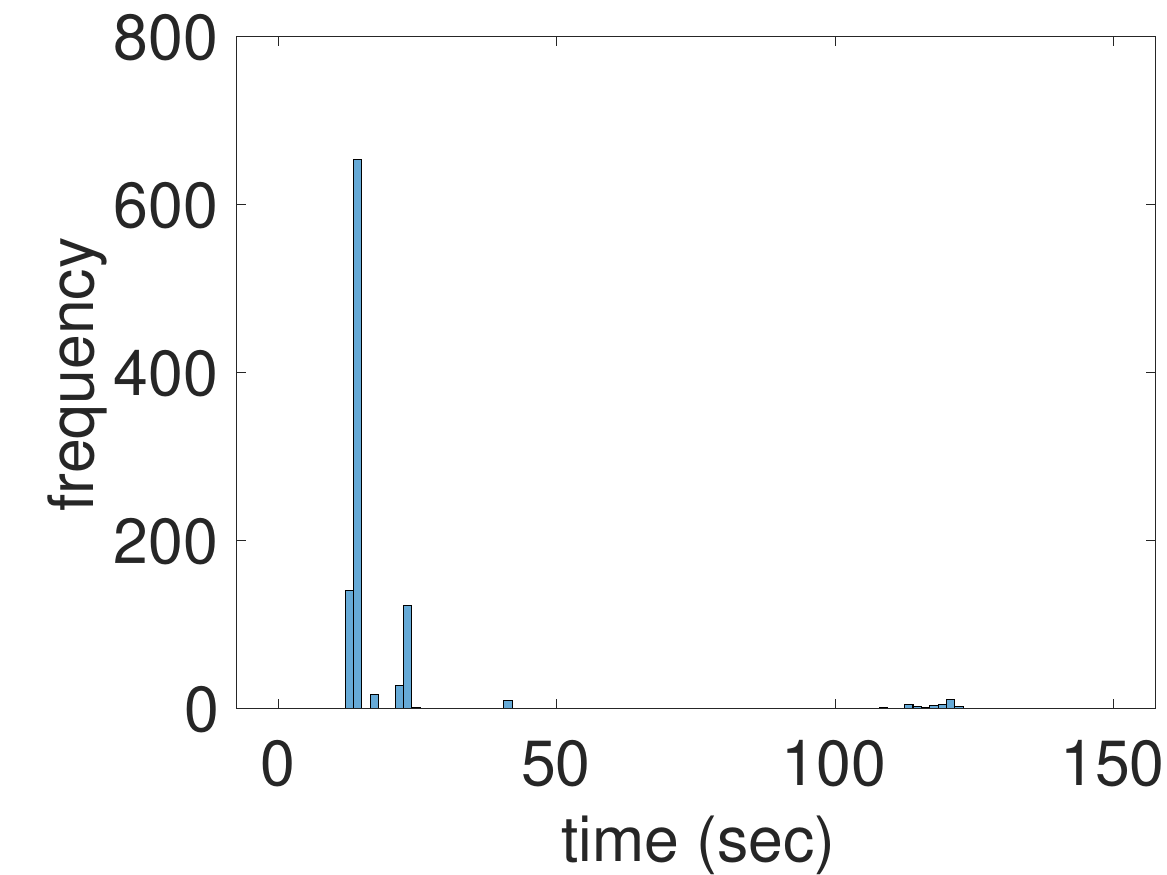}}
  \vspace{-2mm}
  \centerline{\footnotesize (a) Polar, $N=8$}\medskip
\end{minipage}
\hfill
\begin{minipage}[b]{0.28\columnwidth}
  \centering
  \centerline{\includegraphics[width=\textwidth]{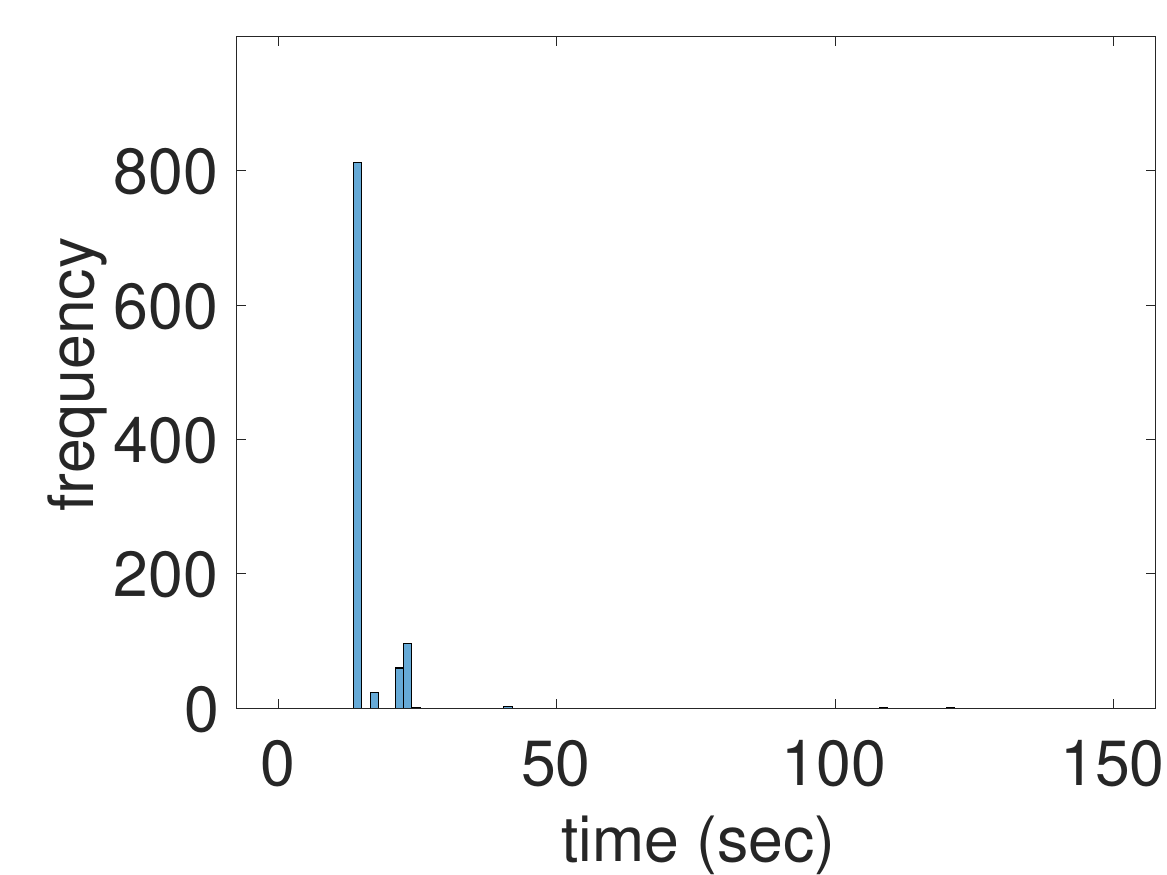}}
  \vspace{-2mm}
  \centerline{\footnotesize (b) Polar, $N=64$}\medskip
\end{minipage}
\hfill
\begin{minipage}[b]{0.28\columnwidth}
  \centering
  \centerline{\includegraphics[width=\textwidth]{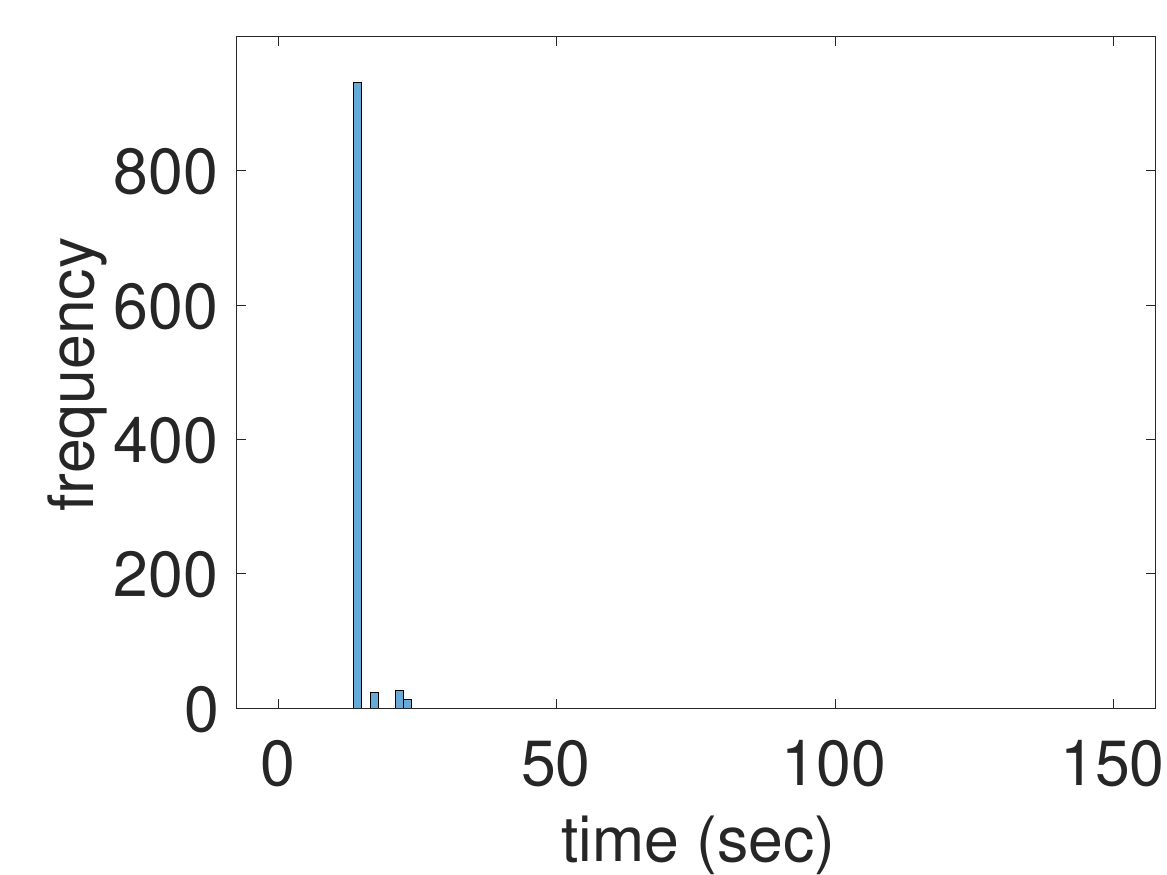}}
  \vspace{-2mm}
  \centerline{\footnotesize (c) Polar, $N=512$}\medskip
\end{minipage}
\hfill
\begin{minipage}[b]{0.28\columnwidth}
  \centering
  \centerline{\includegraphics[width=\textwidth]{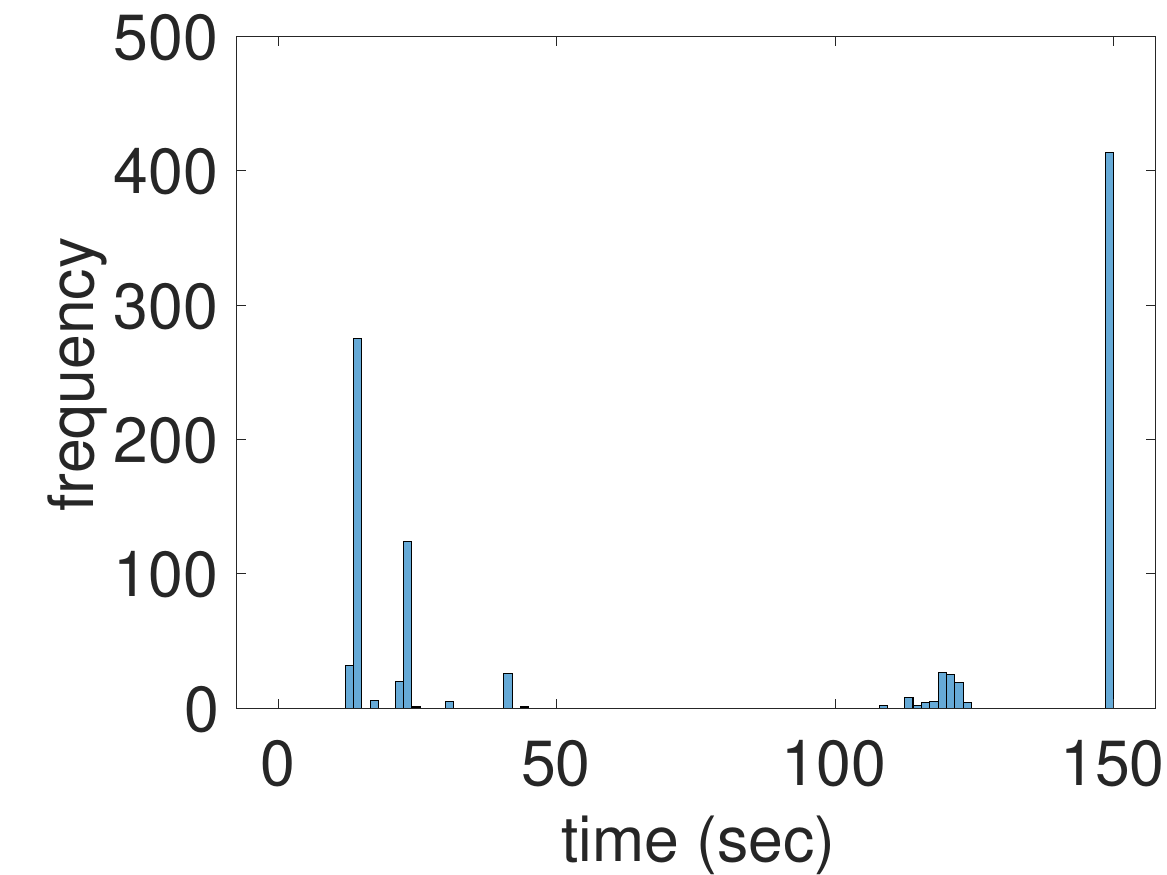}}
  \vspace{-2mm}
  \centerline{\footnotesize (d) LT, $N=8$}\medskip
\end{minipage}
\hfill
\begin{minipage}[b]{0.28\columnwidth}
  \centering
  \centerline{\includegraphics[width=\textwidth]{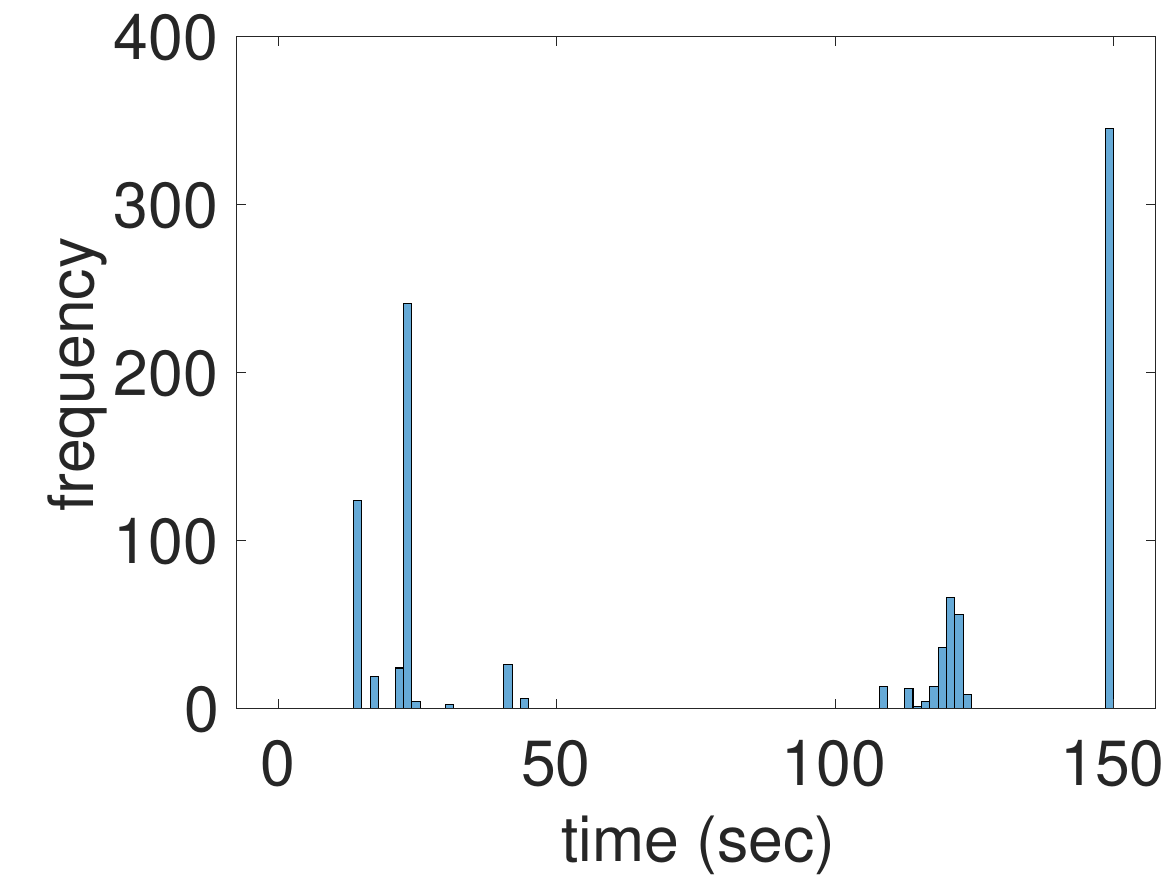}}
  \vspace{-2mm}
  \centerline{\footnotesize (e) LT, $N=64$}\medskip
\end{minipage}
\hfill
\begin{minipage}[b]{0.28\columnwidth}
  \centering
  \centerline{\includegraphics[width=\textwidth]{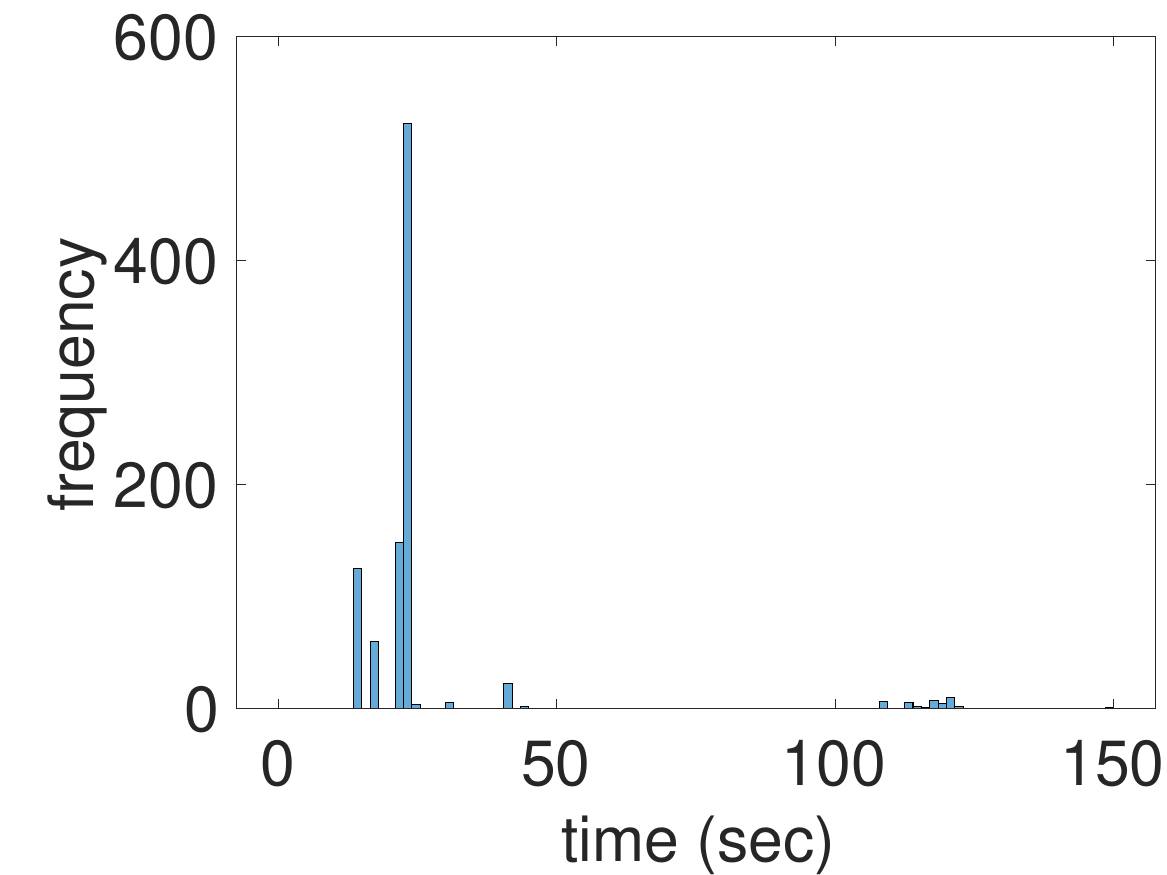}}
  \vspace{-2mm}
  \centerline{\footnotesize (f) LT, $N=512$}\medskip
\end{minipage}
\hfill
\begin{minipage}[b]{0.28\columnwidth}
  \centering
  \centerline{\includegraphics[width=\textwidth]{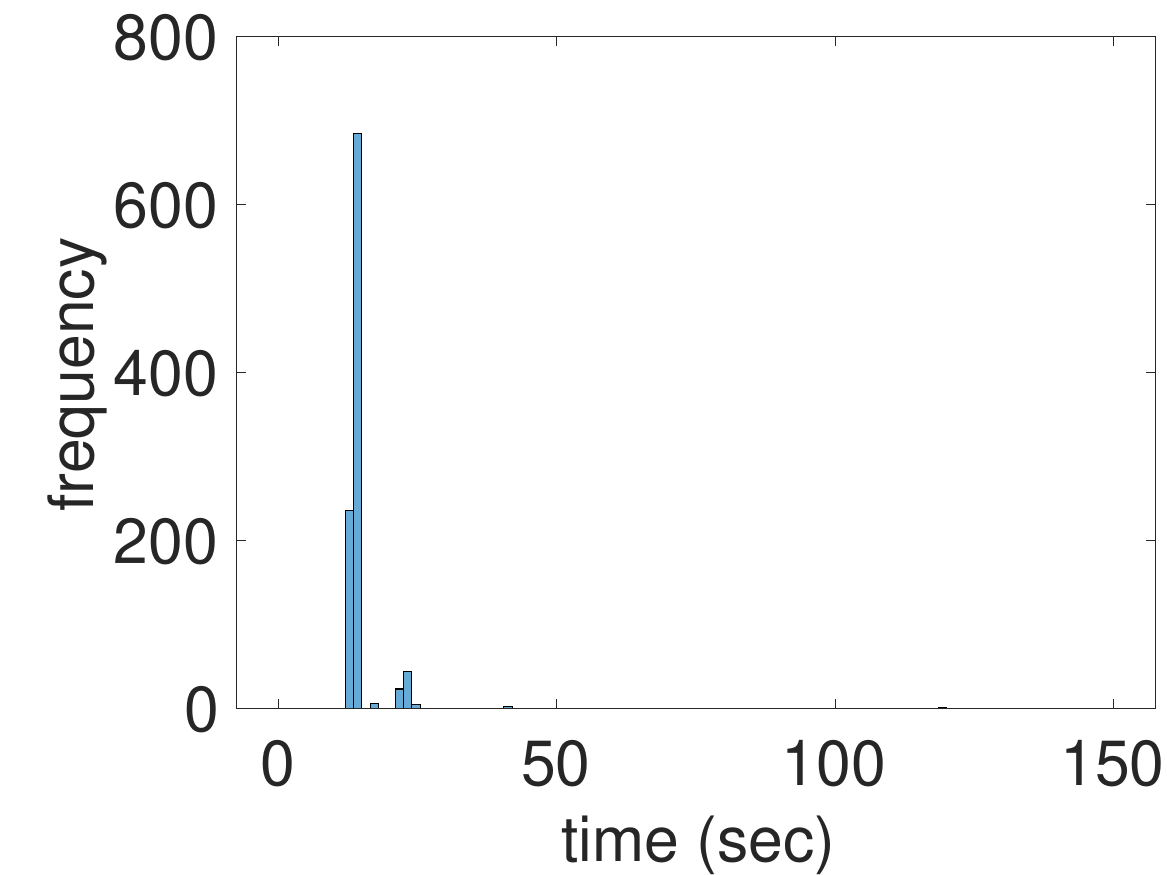}}
  \vspace{-2mm}
  \centerline{\footnotesize (g) MDS, $N=8$}\medskip
\end{minipage}
\hfill
\begin{minipage}[b]{0.28\columnwidth}
  \centering
  \centerline{\includegraphics[width=\textwidth]{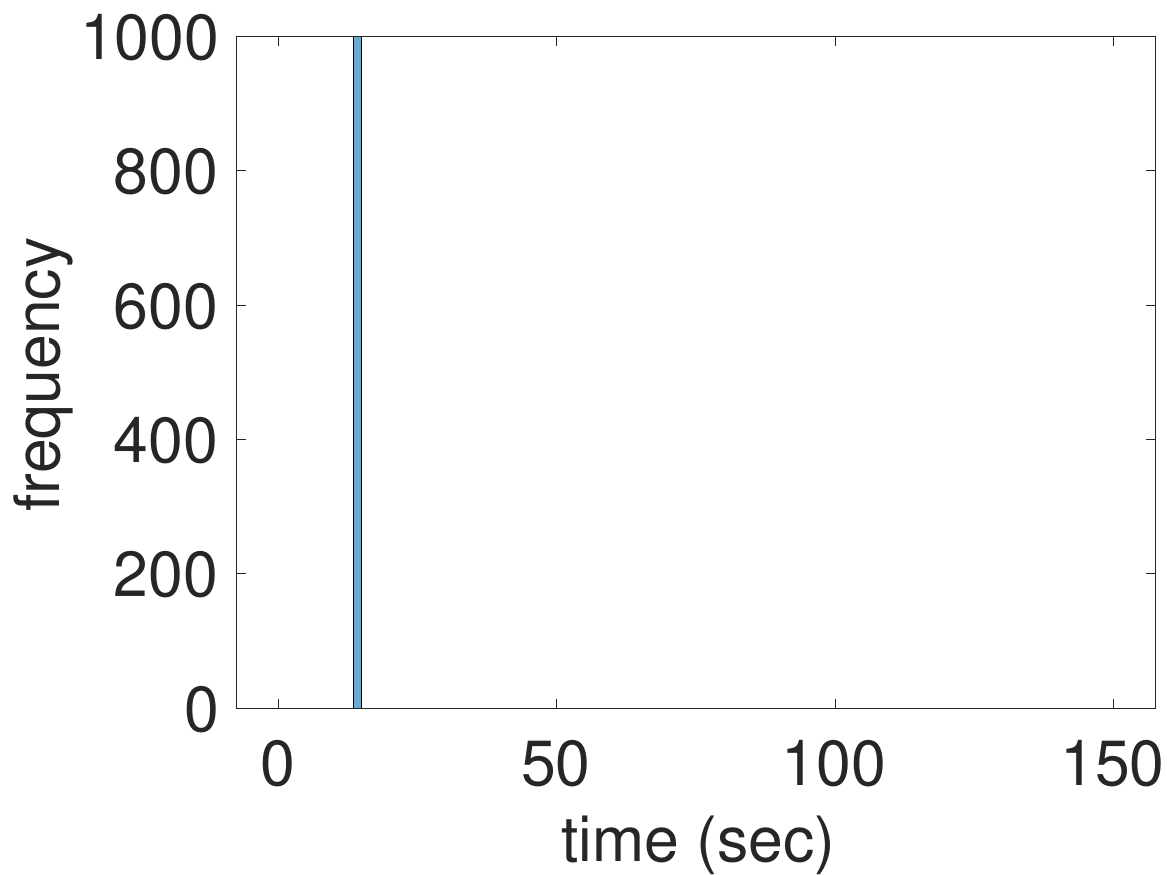}}
  \vspace{-2mm}
  \centerline{\footnotesize (h) MDS, $N=64$}\medskip
\end{minipage}
\hfill
\begin{minipage}[b]{0.28\columnwidth}
  \centering
  \centerline{\includegraphics[width=\textwidth]{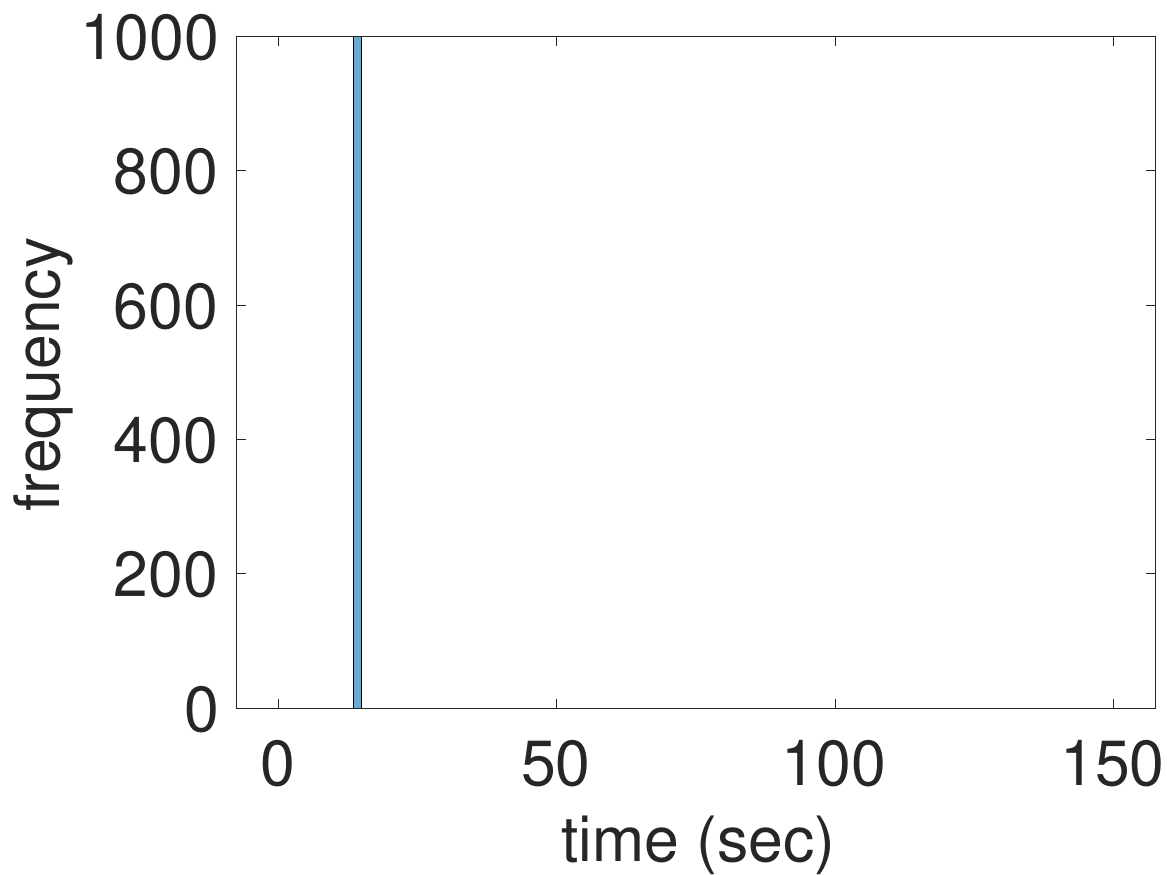}}
  \vspace{-2mm}
  \centerline{\footnotesize (i) MDS, $N=512$}\medskip
\end{minipage}
\vspace{-5mm}
\caption{Histograms of decodability time for polar, LT, and MDS codes.}
\label{decod_time_histograms}
\end{figure}
\section{Conclusion} \label{sec:conclusion}

We have introduced a method that combines polar coding based coded computation and randomized sketching algorithms. Due to their low complexity and simple encoding and decoding algorithms, polar codes lead to favorable run time performance in coded computation. It is critical to note that a low complexity decoder is particularly advantageous when addressing large-scale datasets and a high number of nodes.

A number of works in coded computation literature employ MDS codes for inserting redundancy into computations. In the case where one wishes to work with full-precision data and use Reed-Solomon codes, decoding requires solving a linear system. This will require cubic complexity and lead to unstable solutions for systems with high number of nodes since we are solving a Vandermonde based linear system.
Furthermore, there are many works that restrict their schemes to working with values from a finite field of some size $q$. In that case, it is possible to use fast decoding algorithms which are based on fast algorithms for polynomial interpolation. One such decoding algorithm is given in \cite{soro2009rs_decoder}, which provides $O(N\log N)$ encoding and decoding algorithms for Reed-Solomon erasure codes based on Fermat Number Transform (FNT). The complexity for the encoder is the same as taking a single FNT transform and for the decoder, it is equal to taking $8$ FNT transforms. To compare with polar codes, polar codes require $N\log N$ operations for both encoding and decoding. Another work where a fast erasure decoder for Reed-Solomon codes is presented is \cite{didier2009rs_decoder} which presents a decoder that works in time $O(N\log^2 N)$.

There exist many other fast decoding algorithms for RS codes with complexities as low as $O(N\log N)$. However, the decoding process in these algorithms usually requires taking a fast transform (e.g. FNT) many times and are limited to finite fields. Often these fast decoding algorithms have large hidden constants in their complexity and hence quadratic time decoding algorithms are sometimes preferred over them. Polar codes, on the other hand, provide very straightforward and computationally inexpensive encoding and decoding algorithms. One of our contributions is the design of an efficient decoder for polar codes tailored for the erasure channel that can decode full-precision data.

\bibliographystyle{IEEEtran}
\bibliography{refs}

\begin{thebibliography}{10}
\providecommand{\url}[1]{#1}
\csname url@samestyle\endcsname
\providecommand{\newblock}{\relax}
\providecommand{\bibinfo}[2]{#2}
\providecommand{\BIBentrySTDinterwordspacing}{\spaceskip=0pt\relax}
\providecommand{\BIBentryALTinterwordstretchfactor}{4}
\providecommand{\BIBentryALTinterwordspacing}{\spaceskip=\fontdimen2\font plus
\BIBentryALTinterwordstretchfactor\fontdimen3\font minus
  \fontdimen4\font\relax}
\providecommand{\BIBforeignlanguage}[2]{{%
\expandafter\ifx\csname l@#1\endcsname\relax
\typeout{** WARNING: IEEEtran.bst: No hyphenation pattern has been}%
\typeout{** loaded for the language `#1'. Using the pattern for}%
\typeout{** the default language instead.}%
\else
\language=\csname l@#1\endcsname
\fi
#2}}
\providecommand{\BIBdecl}{\relax}
\BIBdecl

\bibitem{Lee2018}
K.~Lee, M.~Lam, R.~Pedarsani, D.~Papailiopoulos, and K.~Ramchandran, ``Speeding
  up distributed machine learning using codes,'' \emph{IEEE Transactions on
  Information Theory}, vol.~64, no.~3, pp. 1514--1529, 2018.

\bibitem{baharav2018prodcodes}
T.~Baharav, K.~Lee, O.~Ocal, and K.~Ramchandran, ``Straggler-proofing
  massive-scale distributed matrix multiplication with $d$-dimensional product
  codes,'' \emph{IEEE International Symposium on Information Theory (ISIT)},
  pp. 1993--1997, 2018.

\bibitem{yu2017polycode}
Q.~Yu, M.~A. Maddah-Ali, and A.~S. Avestimehr, ``Polynomial codes: An optimal
  design for high-dimensional coded matrix multiplication,'' \emph{Adv. in
  Neural Info. Proc. Systems (NeurIPS) 30}, pp. 4406--4416, 2017.

\bibitem{dean88anytime}
T.~Dean and M.~Boddy, ``An analysis of time-dependent planning,'' in
  \emph{Proceedings of the Seventh AAAI National Conference on Artificial
  Intelligence}, ser. AAAI'88.\hskip 1em plus 0.5em minus 0.4em\relax AAAI
  Press, 1988, p. 49–54.

\bibitem{bartan2019straggler}
B.~Bartan and M.~Pilanci, ``Straggler resilient serverless computing based on
  polar codes,'' in \emph{2019 57th Annual Allerton Conference on
  Communication, Control, and Computing (Allerton)}, 2019, pp. 276--283.

\bibitem{polar2009arikan}
E.~Arikan, ``Channel polarization: A method for constructing capacity-achieving
  codes for symmetric binary-input memoryless channels,'' \emph{IEEE
  Transactions on Information Theory}, vol.~55, pp. 3051--3073, 2009.

\bibitem{dutta2018optimalrec}
S.~Dutta, M.~Fahim, F.~Haddadpour, H.~Jeong, V.~Cadambe, and P.~Grover, ``On
  the optimal recovery threshold of coded matrix multiplication,'' \emph{IEEE
  Transactions on Information Theory}, vol.~66, no.~1, pp. 278--301, 2020.

\bibitem{yu2018straggler}
Q.~Yu, M.~A. Maddah-Ali, and A.~S. Avestimehr, ``Straggler mitigation in
  distributed matrix multiplication: Fundamental limits and optimal coding,''
  in \emph{2018 IEEE International Symposium on Information Theory (ISIT)},
  2018, pp. 2022--2026.

\bibitem{yu2020entangled}
Q.~Yu and A.~S. Avestimehr, ``Entangled polynomial codes for secure, private,
  and batch distributed matrix multiplication: Breaking the "cubic" barrier,''
  in \emph{2020 IEEE International Symposium on Information Theory (ISIT)},
  2020, pp. 245--250.

\bibitem{esfahanizadeh2022layered}
H.~Esfahanizadeh, A.~Cohen, M.~Médard, and S.~Shamai~Shitz, ``Distributed
  computations with layered resolution,'' in \emph{2022 IEEE 11th Int. Conf. on
  Cloud Networking (CloudNet)}, 2022, pp. 257--261.

\bibitem{yang2019timely}
C.-S. Yang, R.~Pedarsani, and A.~S. Avestimehr, ``Timely coded computing,'' in
  \emph{2019 IEEE International Symposium on Information Theory (ISIT)}, 2019,
  pp. 2798--2802.

\bibitem{jahani2023berrut}
T.~Jahani-Nezhad and M.~A. Maddah-Ali, ``Berrut approximated coded computing:
  Straggler resistance beyond polynomial computing,'' \emph{IEEE Transactions
  on Pattern Analysis and Machine Intelligence}, vol.~45, no.~1, pp. 111--122,
  2023.

\bibitem{wang19gradient_coding}
\BIBentryALTinterwordspacing
S.~Wang, J.~Liu, and N.~Shroff, ``Fundamental limits of approximate gradient
  coding,'' \emph{Proc. ACM Meas. Anal. Comput. Syst.}, vol.~3, no.~3, dec
  2019. [Online]. Available: \url{https://doi.org/10.1145/3366700}
\BIBentrySTDinterwordspacing

\bibitem{wang19coded_linear}
S.~Wang, J.~Liu, N.~Shroff, and P.~Yang, ``Computation efficient coded linear
  transform,'' in \emph{Proc. of the Twenty-Second Int. Conf. on Artificial
  Intelligence and Statistics}, ser. Proceedings of Machine Learning Research,
  vol.~89.\hskip 1em plus 0.5em minus 0.4em\relax PMLR, 16--18 Apr 2019, pp.
  577--585.

\bibitem{shashanka17lowrank}
\BIBentryALTinterwordspacing
S.~Ubaru, A.~Mazumdar, and Y.~Saad, ``Low rank approximation and decomposition
  of large matrices using error correcting codes,'' \emph{IEEE Trans. Inf.
  Theor.}, vol.~63, no.~9, p. 5544–5558, sep 2017. [Online]. Available:
  \url{https://doi.org/10.1109/TIT.2017.2723898}
\BIBentrySTDinterwordspacing

\bibitem{reisizadeh17coded}
A.~Reisizadeh, S.~Prakash, R.~Pedarsani, and A.~S. Avestimehr, ``Coded
  computation over heterogeneous clusters,'' \emph{IEEE Transactions on
  Information Theory}, vol.~65, no.~7, pp. 4227--4242, 2019.

\bibitem{gupta2018oversketch}
\BIBentryALTinterwordspacing
V.~Gupta, S.~Wang, T.~Courtade, and K.~Ramchandran, ``Oversketch: Approximate
  matrix multiplication for the cloud,'' in \emph{2018 IEEE International
  Conference on Big Data (Big Data)}.\hskip 1em plus 0.5em minus 0.4em\relax
  Los Alamitos, CA, USA: IEEE Computer Society, dec 2018, pp. 298--304.
  [Online]. Available:
  \url{https://doi.ieeecomputersociety.org/10.1109/BigData.2018.8622139}
\BIBentrySTDinterwordspacing

\bibitem{pilanci2021comppolarization}
M.~Pilanci, ``Computational polarization: An information-theoretic method for
  resilient computing,'' \emph{IEEE Transactions on Information Theory}, pp.
  1--1, 2021.

\bibitem{severinson2018lt}
A.~Severinson, A.~G. i~Amat, and E.~Rosnes, ``Block-diagonal and lt codes for
  distributed computing with straggling servers,'' \emph{IEEE Transactions on
  Communications}, 2018.

\bibitem{mallick2018rateless}
\BIBentryALTinterwordspacing
A.~Mallick, M.~Chaudhari, U.~Sheth, G.~Palanikumar, and G.~Joshi, ``Rateless
  codes for near-perfect load balancing in distributed matrix-vector
  multiplication,'' \emph{SIGMETRICS Perform. Eval. Rev.}, vol.~48, no.~1, p.
  95–96, jul 2020. [Online]. Available:
  \url{https://doi.org/10.1145/3410048.3410104}
\BIBentrySTDinterwordspacing

\bibitem{miguel2016anytime}
J.~S. Miguel and N.~E. Jerger, ``The anytime automaton,'' in \emph{2016
  ACM/IEEE 43rd Annual International Symposium on Computer Architecture
  (ISCA)}, 2016, pp. 545--557.

\bibitem{ferdinand2016anytime}
N.~S. Ferdinand and S.~C. Draper, ``Anytime coding for distributed
  computation,'' in \emph{2016 54th Annual Allerton Conference on
  Communication, Control, and Computing (Allerton)}, 2016, pp. 954--960.

\bibitem{ferdinand2017anytime}
N.~Ferdinand, B.~Gharachorloo, and S.~C. Draper, ``Anytime exploitation of
  stragglers in synchronous stochastic gradient descent,'' in \emph{2017 16th
  IEEE International Conference on Machine Learning and Applications (ICMLA)},
  2017, pp. 141--146.

\bibitem{carreira2018caseserverless}
J.~Carreira, P.~Fonseca, A.~Tumanov, A.~Zhang, , and R.~Katz, ``A case for
  serverless machine learning,'' \emph{Workshop on Systems for ML and Open
  Source Software at NeurIPS 2018}, 2018.

\bibitem{gupta20utility}
\BIBentryALTinterwordspacing
V.~Gupta, S.~Phade, T.~Courtade, and K.~Ramchandran, ``Utility-based resource
  allocation and pricing for serverless computing,'' 2020. [Online]. Available:
  \url{https://arxiv.org/abs/2008.07793}
\BIBentrySTDinterwordspacing

\bibitem{feng2018serverlesstraining}
L.~{Feng}, P.~{Kudva}, D.~{Da Silva}, and J.~{Hu}, ``Exploring serverless
  computing for neural network training,'' in \emph{2018 IEEE 11th
  International Conference on Cloud Computing (CLOUD)}, July 2018, pp.
  334--341.

\bibitem{wang2019serverlesslearning}
H.~{Wang}, D.~{Niu}, and B.~{Li}, ``Distributed machine learning with a
  serverless architecture,'' in \emph{IEEE INFOCOM 2019 - IEEE Conference on
  Computer Communications}, April 2019, pp. 1288--1296.

\bibitem{gupta20oversketched}
V.~Gupta, S.~Kadhe, T.~Courtade, M.~W. Mahoney, and K.~Ramchandran,
  ``Oversketched newton: Fast convex optimization for serverless systems,'' in
  \emph{2020 IEEE Int. Conf. on Big Data (Big Data)}, 2020, pp. 288--297.

\bibitem{salimans2017es}
T.~{Salimans}, J.~{Ho}, X.~{Chen}, S.~{Sidor}, and I.~{Sutskever}, ``Evolution
  strategies as a scalable alternative to reinforcement learning,'' \emph{arXiv
  preprint arXiv:1703.03864}, 2017.

\bibitem{choromanski2018structured}
K.~Choromanski, M.~Rowland, V.~Sindhwani, R.~Turner, and A.~Weller,
  ``Structured evolution with compact architectures for scalable policy
  optimization,'' in \emph{Proceedings of the 35th International Conference on
  Machine Learning}, ser. Proceedings of Machine Learning Research,
  vol.~80.\hskip 1em plus 0.5em minus 0.4em\relax PMLR, 10--15 Jul 2018, pp.
  970--978.

\bibitem{imagenet}
J.~Deng, W.~Dong, R.~Socher, L.-J. Li, K.~Li, and L.~Fei-Fei, ``Imagenet: A
  large-scale hierarchical image database,'' in \emph{2009 IEEE Conference on
  Computer Vision and Pattern Recognition}, 2009, pp. 248--255.

\bibitem{Mahoney11}
M.~W. Mahoney, ``Randomized algorithms for matrices and data,''
  \emph{Foundations and {T}rends in {M}achine {L}earning in Machine Learning},
  vol.~3, no.~2, 2011.

\bibitem{tropp2011improved}
J.~A. Tropp, ``Improved analysis of the subsampled randomized hadamard
  transform,'' \emph{Advances in Adaptive Data Analysis}, vol.~3, no. 01n02,
  pp. 115--126, 2011.

\bibitem{krahmer2011new}
F.~Krahmer and R.~Ward, ``New and improved johnson--lindenstrauss embeddings
  via the restricted isometry property,'' \emph{SIAM Journal on Mathematical
  Analysis}, vol.~43, no.~3, pp. 1269--1281, 2011.

\bibitem{krizhevsky2012imagenet}
A.~Krizhevsky, I.~Sutskever, and G.~E. Hinton, ``Imagenet classification with
  deep convolutional neural networks,'' in \emph{Advances in neural information
  processing systems}, 2012, pp. 1097--1105.

\bibitem{soro2009rs_decoder}
A.~{Soro} and J.~{Lacan}, ``Fnt-based reed-solomon erasure codes,'' in
  \emph{7th IEEE Consumer Comm. and Networking Conf.}, 2010, pp. 1--5.

\bibitem{didier2009rs_decoder}
\BIBentryALTinterwordspacing
F.~Didier, ``Efficient erasure decoding of reed-solomon codes,'' \emph{CoRR},
  vol. abs/0901.1886, 2009. [Online]. Available:
  \url{http://arxiv.org/abs/0901.1886}
\BIBentrySTDinterwordspacing

\bibitem{Boyd02}
S.~Boyd and L.~Vandenberghe, \emph{Convex optimization}.\hskip 1em plus 0.5em
  minus 0.4em\relax Cambridge, UK: Cambridge University Press, 2004.

\end{thebibliography}

% uncomment below for appendix 
\newpage
\appendices
\setcounter{page}{1}
\onecolumn
\section{Extension to Coded Black-Box Optimization} \label{sec:coded_optimization}

In this section, we extend our coding method to the minimization of $f(\theta)$ where $f : \mathbb{R}^d \to \mathbb{R}$. We assume that we do not have access to the analytical form of $f(\theta)$ or its gradient, and we assume that we are only able to make queries for function evaluations. This setting is called black-box optimization. 

We use $\nabla_v$ to denote directional derivative along the direction $v \in \mathbb{R}^d$. We use $e_i$ to denote the $i$'th unit column vector with the appropriate dimension. $H$ refers to the Hadamard matrix (its dimension can be determined from the context), and $\mathbf{H}$ is the Hessian of a function. We now briefly describe the black-box optimization methods of finite differences and evolution strategies and then present the proposed coded black-box optimization method.

\textbf{Finite Differences:} For black-box optimization problems, approximate gradients can be obtained by using the finite differences estimator. The derivative of a function $f(x)$ with respect to the variable $x_i$ can be approximated by using
\begin{align}\label{finite_diff}
    \frac{\partial f(x)}{\partial x_i} \approx \frac{f(x+\delta e_i) - f(x-\delta e_i)}{2\delta} \,,
\end{align}
where $\delta \in \mathbb{R}$ is a small scalar that determines the perturbation amount. One can obtain an approximate gradient using \eqref{finite_diff} and use it in gradient-based optimization methods such as gradient descent.

The finite differences estimator can easily exploit parallelism since each partial approximate derivative can be independently evaluated in different worker nodes in parallel. This method can be made straggler-resilient by using only the available derivative estimates and ignoring the outputs of the slower workers. However, this typically leads to slower convergence. Another alternative is to replicate finite difference calculations, which is not optimal from a coding theory perspective.

\textbf{Evolution Strategies:} Let us consider the following evolution strategies (ES) gradient estimator (\cite{salimans2017es}, \cite{choromanski2018structured})
\begin{align}\label{ES_estimator}
    \nabla f(\theta) \approx \frac{1}{2N\delta} \sum_{i=1}^{N}(f(\theta+\delta \epsilon_i)\epsilon_i - f(\theta-\delta \epsilon_i)\epsilon_i) \,,
\end{align}
where $\delta$ is the scaling coefficient for the random perturbation directions $\epsilon_i$ and $N$ is the number of perturbations. This estimator is referred to as antithetic evolution strategies gradient estimator. The random perturbation directions $\epsilon_i$ may be sampled from a standard multivariate Gaussian distribution $\mathcal{N}(0,I)$. Alternatively, $\epsilon_i$'s may be generated as the rows of the matrix $HD$ where $H$ is the Hadamard matrix and $D$ is a diagonal matrix with entries distributed as Rademacher distribution.

It is shown in \cite{choromanski2018structured} that if exploration (or perturbation) directions $\epsilon_i$ are orthogonal, the gradient estimators lead to a lower error. We note that the rows of $HD$ are orthogonal with the appropriate scaling factor. We omit this scaling factor by absorbing it in the $\delta$ term which scales perturbation directions $\epsilon_i$. This work considers the case where random perturbation directions are generated according to $HD$ and we will refer to it as \textit{structured evolution strategies} as it is done in \cite{choromanski2018structured}.

Evolution strategies can exploit parallelism as well since workers need to communicate only scalars which are the function evaluations and the random seeds used when generating the random perturbation directions.

%%%%%%%%%%%%%%%%%%%%%%%%%%%%%%%%%%%
\subsection{Distributed Black-Box Optimization using Polar Codes}

We now present the proposed method for speeding up distributed black-box optimization. We start by introducing some more notation and definitions. The derivative of a differentiable function $f$ at a point $\theta$ along the unit vector direction $e_i$ is the $i$'th component of the gradient $\nabla f$, that is, $
    \frac{\partial f}{\partial \theta_i} = e_i^T \nabla f.$
The directional derivative along $v$ is defined as follows
\begin{align}\label{direc_deriv_2}
    \nabla_{v}f = \lim_{\delta \rightarrow 0} \frac{f(\theta + \delta v) - f(\theta) }{\delta}\,.
\end{align}
If the function is differentiable at a point $\theta$, then the directional derivative exists along any direction $v$ and is a linear map \cite{Boyd02}. In this case we have
\begin{align}\label{direc_deriv_1}
    \nabla_{v}f = v^T \nabla f.
\end{align}
When we do not have access to exact gradients, we can employ a \emph{numerical directional derivative} by choosing a small $\delta$ in
\begin{align}\label{direc_deriv_approx}
    \nabla_{v}f \approx \frac{f(\theta + \delta v) - f(\theta) }{\delta} \,.
\end{align}
Note that the approximation in \eqref{direc_deriv_approx} is not symmetric, that is, it involves perturbing the parameters only along $+v$. We instead use the symmetric version of \eqref{direc_deriv_approx} for approximating derivatives in which the parameters are perturbed along both the directions $-v$ and $+v$:
\begin{align}\label{direc_deriv_symmetric}
    \nabla_{v}f \approx \frac{f(\theta + \delta v) - f(\theta - \delta v) }{2\delta}\,.
\end{align}
If we consider the Taylor series expansion for $f(\theta + \delta v)$ and $f(\theta - \delta v)$, we obtain
\begin{align}
    f(\theta + \delta v) &= f(\theta) + \delta \nabla f^T v + \frac{\delta^2}{2} v^T \mathbf{H}v + O(\delta^3) \nonumber \\
    f(\theta - \delta v) &= f(\theta) - \delta \nabla f^T v + \frac{\delta^2}{2} v^T \mathbf{H}v + O(\delta^3)
\end{align}
where $O(\delta^3)$ is a third order error term and $\mathbf{H}$ is the Hessian matrix for $f$. Substituting these expansions in \eqref{direc_deriv_symmetric}, we obtain
\begin{align}\label{direc_deriv_final}
    \nabla_{v}f &\approx \frac{ 2\delta \nabla f^T v + O(\delta^3) }{2\delta} = \nabla f^T v + O(\delta^2) \,.
\end{align}
This shows that for small $\delta$, the numerical directional derivative becomes approximately linear in $\nabla f$. Our proposed method makes use of this assumption that directional derivative estimates are approximately linear in the directions $v$ to ensure that coding can be applied to directional derivative estimates. To make this more concrete, let us consider the construction given in Figure \ref{example_construction}. The block $H$ corresponds to the Hadamard kernel $H=\bigl[\begin{smallmatrix} 1 & 1 \\ 1 & -1 \end{smallmatrix}\bigr]$.

\begin{figure}
  \centering
  \includegraphics[width=0.51\textwidth]{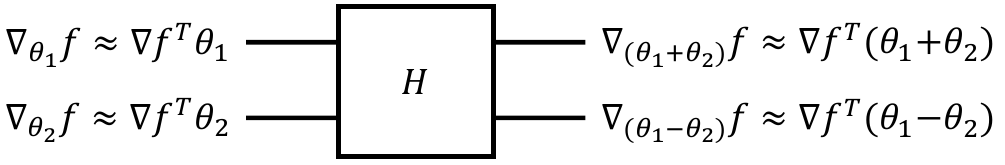}
  \caption{2-by-2 construction based on Hadamard transformation.}
  \label{example_construction}
\end{figure}

In Figure \ref{example_construction}, if we know the estimates for $\nabla_{(\theta_1+\theta_2)}f$ and $\nabla_{(\theta_1-\theta_2)}f$, we can compute the estimates for $\nabla_{\theta_1}f$ and $\nabla_{\theta_2}f$ because the directional derivative estimates are approximately linear in their corresponding directions. Furthermore, if we know the estimate for $\nabla_{\theta_1}f$, then it is sufficient to know only one of the estimates for $\nabla_{(\theta_1+\theta_2)}f$ or $\nabla_{(\theta_1-\theta_2)}f$ in order to be able to compute the estimate for $\nabla_{\theta_2}f$. This would happen if, for example, $\theta_1$ is the zero vector (i.e. frozen direction) because the estimate for $\nabla_{\theta_1}f$ would be zero and either of the directional derivative estimates from the right-hand side would be enough for us to obtain the estimate of $\nabla_{\theta_2}f$.
We refer to the directions $\theta_1$, $\theta_2$, $(\theta_1+\theta_2)$, $(\theta_1-\theta_2)$ as perturbation directions.

We now summarize the proposed method and the descriptions of the steps will follow.
\begin{itemize}
    \item Encode all the unit vectors in $\mathbb{R}^d$ to obtain the encoded perturbation directions.
    % \vspace{-2mm}
    \item Assign each perturbation direction to a worker node and have them compute their directional derivative estimates using \eqref{direc_deriv_symmetric}.
    % \vspace{-2mm}
    \item Central node starts collecting worker outputs. 
    % \vspace{-2mm}
    \item When a decodable set of worker outputs is available, the central node decodes these outputs to obtain an estimate for the gradient.
    % \vspace{-2mm}
    \item The central node computes the next iterate for the parameter $\theta$.
    % \vspace{-2mm}
    \item Repeat until convergence or for a desired number of iterations.
\end{itemize}

The above procedure assumes that we want to estimate all entries of the gradient, but it is possible estimate only a portion of gradient entries by encoding only the unit vectors corresponding to the desired entries. 
Moreover, we note that we can always check whether decoding helps in obtaining a better objective function compared to the the structured evolution strategies and make the update accordingly. Since the decoding step is fast due to the efficient polar decoder, decoding the outputs but not using the recovered estimate does not place a heavy computational burden.

\subsection{Encoding}
Encoding is computed based on the Hadamard transformation whose kernel is shown in Figure \ref{two_by_two_perturb}.
\begin{figure}%[htbp]
  \centering
  \includegraphics[width=0.31\textwidth]{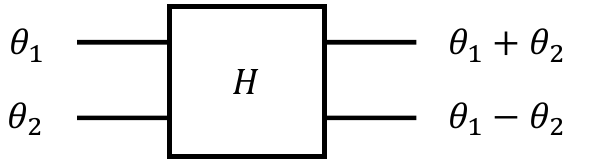}
  \caption{2-by-2 Hadamard transformation of the perturbation directions.}
  \label{two_by_two_perturb}
\end{figure}
In channel coding, freezing channels corresponds to sending known bits, e.g., the zero bit. Here, freezing inputs corresponds to setting them to all-zero coordinates so that the corresponding directional derivative is zero. This makes it possible, when decoding, to take the value of the derivative estimates for frozen nodes to be zero. For the information nodes, we simply send in unit vectors.
For instance, encoding for a function of $3$ variables using $N=4$ workers is shown in Figure \ref{example_encoding}. The resulting $4$ output vectors are the perturbation directions.

\begin{figure}%[htbp]
  \centering
  \includegraphics[width=0.6\textwidth]{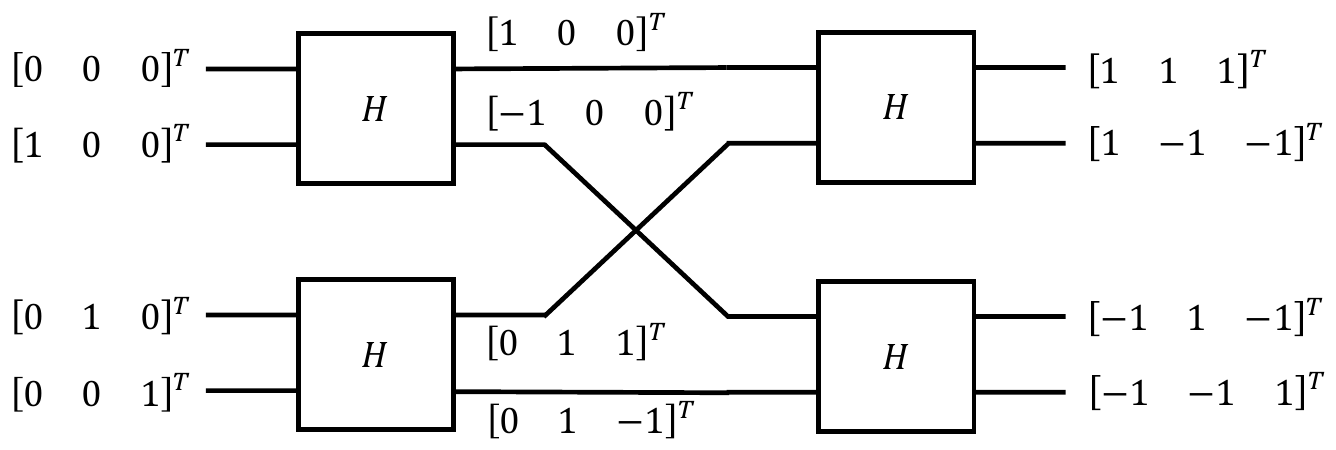}
  \caption{Example encoding for $N=4$}
  \label{example_encoding}
\end{figure}

In this construction, the rate is $3/4$ since $3$ out of $4$ inputs are used for sending in unit vectors. When the erasure probabilities of the transformed nodes are computed (see \cite{polar2009arikan}), one will see that the worst one corresponds to the first index. Hence, the first input is frozen and the remaining $3$ inputs are the information nodes. The frozen input is set to the zero vector and the others are set to $3$-dimensional unit vectors. Note that for the frozen input, perturbing by the zero vector is the same as not perturbing the variables and hence we get $f(\theta + 0) - f(\theta - 0) = 0$. This is important since during decoding, we will not have to do any computations to evaluate the value of the frozen inputs as we know they are zero. 

If we wish to use $N=8$ workers instead of $4$, we would set the worst $5$ inputs to zero vectors and the remaining best three inputs would be set to the unit vectors $e_1, e_2, e_3$. The rate in this case would be $3/8$, and this construction would be more straggler-resilient since we can recover the gradient estimate in the presence of even more stragglers compared to the $N=4$ case.

\textbf{Embedding interpretation:} It is also possible to perform the encoding step slightly differently for a different view on freezing channels. Instead of setting frozen inputs to zero vectors, one can increase the dimension of the inputs from $d$ to $N$ and set the frozen channels to unit vectors $e_j$ where $j \in \{d+1,d+2,\dots,N\}$. Because the function $f(\theta)$ accepts $d$-dimensional inputs, we could embed $f(\theta)$ into a higher dimension, that is, we could define $\tilde{f}(\tilde{\theta})$ where $\tilde{\theta} \in \mathbb{R}^N$ and $\tilde{f}(\tilde{\theta}) = f(\theta)$ if $\tilde{\theta}_i=\theta_i$ for $i \in \{1,2,\dots,d\}$. The advantage of this approach is that the output of the encoding step will be equal to the $H$ matrix with permuted rows.

\subsection{Decoding}
The sequential decoder given in Algorithm \ref{decoding_alg} directly handles linear operations such as matrix-vector multiplication with full-precision data (i.e. does not require finite field data). Since the gradient estimates $\nabla_{h_i}f$ can be linearly approximated, we can use the same decoding method.

We note that the structured evolution strategies method is based on exploring the parameter space along the rows of $HD$ instead of $H$ alone. So far, we have only considered the $H$ matrix for perturbation directions. It is possible to incorporate the diagonal matrix $D$ into our method as well.
Multiplying $H$ by $D$ from the right corresponds to multiplying all the entries of the $i$'th column of $H$ by $D_{i}$ for all $i$. In this case, instead of computing an estimate for the directional derivative along a direction of $v$, we  approximate the directional derivative with respect to the direction $Dv$.
\section{Proofs}
\begin{proof}[Proof of Lemma \ref{polarizing_kernel_lemmma}]
We first prove that if $K$ is a polarizing kernel, then it satisfies both of the given conditions. Let us assume an $f$ function satisfying the linearity property given in Definition \ref{polarizing_kernel_def}. Since $f$ satisfies the linearity property, we can write
\begin{align}
    \left[\begin{matrix} f(v_1) \\ f(v_2) \end{matrix}\right] = \left[\begin{matrix}K_{11} & K_{12}\\ K_{21} & K_{22} \end{matrix}\right] \times \left[\begin{matrix} f(u_1) \\ f(u_2) \end{matrix}\right].
\end{align} 
Computing $f(u_2)$ given the value of $f(u_1)$ in time $\min(T_1,T_2)$ means that $f(u_2)$ can be computed using $f(u_1)$ and either one of $f(v_1)$, $f(v_2)$ (whichever is computed earlier). This implies that we must be able to recover $f(u_2)$ using one of the following two equations
\begin{align} 
    K_{12} \times f(u_2) &= f(v_1) - K_{11}f(u_1) \label{lemma_eq_1} \\
    K_{22} \times f(u_2) &= f(v_2) - K_{21}f(u_1). \label{lemma_eq_2}
\end{align}
We use \eqref{lemma_eq_1} if $f(v_1)$ is known, and \eqref{lemma_eq_2} if $f(v_2)$ is known. This implies that both $K_{12}$ and $K_{22}$ need to be non-zero. Furthermore, to be able to compute $f(u_1)$ in time $\max(T_1, T_2)$ means that it is possible to find $f(u_1)$ using both $f(v_1)$ and $f(v_2)$ (note that we do not assume we know the value of $f(u_2)$). There are two scenarios where this is possible: Either at least one row of $K$ must have its first element as non-zero and its second element as zero, or $K$ must be invertible. Since we already found out that $K_{12}$ and $K_{22}$ are both non-zero, we are left with one scenario, which is that $K$ must be invertible.

We proceed to prove the other direction of the `if and only if' statement, which states that if a kernel $K \in \mathbb{R}^{2 \times 2}$ satisfies the given two conditions, then it is a polarizing kernel. We start by assuming an invertible $K \in \mathbb{R}^{2 \times 2}$ with both elements in its second column non-zero. Since $K$ is invertible, we can uniquely determine $f(u_1)$ when both $f(v_1)$ and $f(v_2)$ are available, which occurs at time $\max(T_1,T_2)$. Furthermore, assume we know the value of $f(u_1)$. At time $\min(T_1,T_2)$, we will also know one of $f(v_1)$, $f(v_2)$, whichever is computed earlier. Knowing $f(u_1)$, and any one of $f(v_1)$, $f(v_2)$, we can determine $f(u_2)$ using the suitable one of the equations \eqref{lemma_eq_1}, \eqref{lemma_eq_2} because $K_{12}$ and $K_{22}$ are both assumed to be non-zero. Hence this completes the proof that a kernel $K$ satisfying the given two conditions is a polarizing kernel.
\end{proof}

%%%%%%%%%%%%%%%%%%%%
\begin{proof}[Proof of Theorem \ref{thm_decoder_kernel}]
By Lemma \ref{polarizing_kernel_lemmma}, we know that for $K$ to be a polarizing kernel, it must be invertible and must have both $K_{12}$ and $K_{22}$ as non-zero. For a $2\times 2$ matrix to be invertible with both second column elements as non-zero, at least one of the elements in the first column must also be non-zero. We now know that $K_{12}$, $K_{22}$ and at least one of $K_{11}$, $K_{21}$ must be non-zero in a polarizing kernel $K$. It is easy to see that having all four elements of $K$ as non-zero leads to more computations than having only three elements of $K$ as non-zero. Hence, we must choose either $K_{11}$ or $K_{21}$ to be zero (it does not matter which one). It is possible to avoid any multiplications by selecting the non-zero elements of $K$ as ones. Hence, both $K=\bigl[\begin{smallmatrix}1 & 1\\ 0 & 1\end{smallmatrix}\bigr]$ and $K=\bigl[\begin{smallmatrix}0 & 1\\ 1 & 1\end{smallmatrix}\bigr]$ are polarizing kernels and lead to the same amount of computations, which is a single addition. This amount of computations is the minimum possible as otherwise $K$ will not satisfy the condition that a polarizing kernel must have at least $3$ non-zero elements.
\end{proof}

\section{Partial Construction}
In this section, we introduce a novel idea that we refer to as \textit{partial construction} to scale up the encoding procedure. Suppose that we are interested in computing the linear operation $Ax$ where we only encode $A$ and not $x$. If the data matrix $A$ is extremely large, it may be time consuming to encode the data. A way around having to encode a large $A$ is to consider partial code constructions, that is, for $A=[A_1^T, \dots, A_p^T]^T$, we encode each submatrix $A_i$ separately.  The encoded $A$ can be written as follows: 
\begin{align}
    \mathrm{encode}(A) = \begin{bmatrix}
    (Z \otimes I_{n/(ps)} )A_1 \\ \vdots \\ (Z \otimes I_{n/(ps)} )A_p
\end{bmatrix}.
\end{align}

It follows that decoding the outputs of the construction for the submatrix $A_i$ will give us $A_ix$. This results in a weaker straggler resilience, however, we get a trade-off between the computational load of encoding and straggler resilience. Partial construction also decreases the amount of computations required for decoding since instead of decoding a code with $N$ outputs (of complexity $O(N\log N)$), now we need to decode $p$ codes with $\frac{N}{p}$ outputs, which is of complexity $O(N\log (\frac{N}{p}))$.

In addition, partial construction makes it possible to parallel compute both encoding and decoding. Each code construction can be encoded and decoded independently from the rest of the constructions. Partial construction idea can also be applied to coded computation schemes based on other codes. For instance, one scenario where this idea is useful is when one is interested in using RS codes with full-precision data. Given that for large $N$ values, using RS codes with full-precision data becomes infeasible, one can construct many smaller size codes. When the code size is small enough, a Vandermonde-based linear system can be painlessly solved.

Another benefit of the partial construction idea is that for constructions of sizes small enough, encoding can be performed in the memory of the workers after reading the necessary data. This results in a straggler-resilient scheme without doing any pre-computing to encode the entire dataset. In-memory encoding could be useful for problems where the data matrix $A$ is changing over time because it might be too expensive to encode the entire dataset $A$ every time it gets updated.

\section{Coded Matrix Multiplication} \label{sec:mat_mat_mult}
In this section, we provide an extension to our proposed method to accommodate coding of both $A$ and $B$ for computing the matrix multiplication $AB$ (instead of coding only $A$). This can be thought of as a two-dimensional extension of our method. Let
$A = [ A_1^T, \dots ,A_{d_1}^T ]^T$ and $ B = [B_1 ,\dots, B_{d_2} ]$. Let us denote zero matrix padded version of $A$ by $\tilde{A}=[\tilde{A}_1^T, \dots, \tilde{A}_{N_1}^T]^T$ such that $\tilde{A}_i = 0$ if $i$ is a frozen channel index and $\tilde{A}_i = A_j$ if $i$ is a data channel index with $j$ the appropriate index. Similarly, we define $\tilde{B}=[\tilde{B}_1, \dots, \tilde{B}_{N_2}]$ such that $\tilde{B}_i = 0$ if $i$ is a frozen channel index and $\tilde{B}_i =B_j$ if $i$ is a data channel index with $j$ the appropriate index. Encoding on $\tilde{A}$ can be represented as $G_{N_1} \tilde{A}$ where $G_{N_1}$ is the $N_1$ dimensional generator matrix and acts on submatrices $\tilde{A}_i$. Similarly, encoding on $\tilde{B}$ would be $\tilde{B} G_{N_2}$.

Encoding $A$ and $B$ gives us  $N_1$ submatrices $(G_{N_1} \tilde{A})_i$ and $N_2$ submatrices $(\tilde{B} G_{N_2})_j$. We multiply the encoded matrices using $N_1 N_2$ workers with the $(i,j)$th worker computing the multiplication $(G_{N_1} \tilde{A})_i (\tilde{B} G_{N_2})_j$. So, the worker outputs will be of the form:

\begin{align}\label{definition_of_P}
    P = \left[\begin{matrix} (G_{N_1} \tilde{A})_1 (\tilde{B} G_{N_2})_1 & \hdots & (G_{N_1} \tilde{A})_1 (\tilde{B} G_{N_2})_{N_2} \\
    \vdots & \ddots & \\ 
    (G_{N_1} \tilde{A})_{N_1} (\tilde{B} G_{N_2})_1 & \hdots & (G_{N_1} \tilde{A})_{N_1} (\tilde{B} G_{N_2})_{N_2} \end{matrix}\right].
\end{align}
Note that for fixed $j$, the worker outputs are:

\begin{align}
    \left[\begin{matrix} (G_{N_1} \tilde{A})_1 (\tilde{B} G_{N_2})_j  \\ \vdots \\ (G_{N_1} \tilde{A})_{N_1} (\tilde{B} G_{N_2})_j \end{matrix}\right].
\end{align}
For fixed $j$, the outputs are linear in $(\tilde{B} G_{N_2})_j$, hence, it is possible to decode these outputs using the decoder we have for the 1D case. Similarly, for fixed $i$, the outputs are:

\begin{align}
    \left[\begin{matrix} (G_{N_1} \tilde{A})_i (\tilde{B} G_{N_2})_1 & \hdots & (G_{N_1} \tilde{A})_i (\tilde{B} G_{N_2})_{N_2} \end{matrix}\right].
\end{align}
For fixed $i$, the outputs are linear in $(G_{N_1} \tilde{A})_i$. It follows that the 1D decoding algorithm can be used for decoding the outputs. Based on these observations, the decoder algorithm for the 2D case is given in Algorithm \ref{2d_decoder}. The 2D decoding algorithm makes calls to the 1D encoding and decoding algorithms many times to fill in the missing entries of the encoded matrix $P$ defined in \eqref{definition_of_P}. When all missing entries of $P$ are computed, first all rows and then all columns of $P$ are decoded and finally, the frozen entries are removed to obtain the multiplication $AB$.

%\Comment*[r]{Part I}}
\begin{algorithm}
 \KwIn{the worker output matrix $P$}
 \KwResult{$y = A\times B$}
 
 \While{$P$ has missing entries}{
 \Comment*[r]{loop over rows}
 \vspace{-0.375cm}
  \For{$i \gets 0$ \textbf{to} $N_1-1$ } {
    \If{$P[i,:]$ has missing entries and is decodable}{
     decode $P[i,:]$ using Alg. \ref{decoding_alg} \\
     forward propagation to fill in $P[i,:]$ 
    }
  }
  \Comment*[r]{loop over columns}
  \vspace{-0.375cm}
  \For{$j \gets 0$ \textbf{to} $N_2-1$} {
    \If{$P[:,j]$ has missing entries and is decodable}{
     decode $P[:,j]$ using Alg. \ref{decoding_alg} \\
     forward propagation to fill in $P[:,j]$ 
    }
  }
 }
 decode all rows and then all columns of $P$ \\
 return entries of $P$ (ignoring the frozen entries)
 \caption{2D decoding algorithm.}
 \label{2d_decoder}
\end{algorithm}

\section{Privacy} \label{subsec:privacy}

Worker nodes may be unreliable in distributed computing and in such cases it is desirable to introduce privacy into the computation. It is possible to incorporate privacy into coded computation with polar codes in a very straightforward way. For simplicity, let us assume that no two workers can collude. Then, we note that selecting the last input as a random matrix $R$ in the code construction, that is, $U_N=R$, leads to all worker outputs containing the product $Rx$ as an additive term. Hence, this results in a privacy-preserving computation against single-node attacks.

\end{document}